\documentclass[10pt]{article} 
\usepackage[a4paper]{geometry}
\usepackage[utf8]{inputenc} 

\usepackage{graphics} 
\usepackage{epsfig} 
\usepackage{subcaption} 
\usepackage{caption}
\usepackage{amsmath} 
\usepackage{amssymb}  
\usepackage{amsthm}
\usepackage{natbib}

 \usepackage{bbm}
\usepackage{algorithm}
\usepackage{algpseudocode}

\usepackage{pgf-pie}
\usepackage{pgfplots}
\usepackage{tikz}
\usetikzlibrary{matrix,positioning}
\tikzset{bullet/.style={circle,fill,inner sep=2pt}}

\newtheorem{lemma}{Lemma}[section]

\newtheorem{remark}{Remark}[section]

\newtheorem{example}{Example}[section]

\usepackage{verbatim}
\usepackage{graphicx} 
\usepackage{epstopdf}
\usepackage{color}
\usepackage{multirow}
\usepackage{lipsum}
\makeatletter
\def\footnoterule{\relax%
	\kern-5pt
	\hbox to \columnwidth{\hfill\vrule width .9\columnwidth height 0.4pt\hfill}
	\kern4.6pt}
\makeatother

\usepackage[appendix = inline]{apxproof}
\newtheoremrep{theorem}{Theorem}[section]
\newtheoremrep{lemma}{Lemma}[section]
\newtheoremrep{corollary}{Corollary}[section]

\usepackage{color,hyperref}
\definecolor{darkblue}{rgb}{0.0,0.0,0.6}
\hypersetup{colorlinks,breaklinks,linkcolor=darkblue,urlcolor=darkblue,anchorcolor=darkblue,citecolor=darkblue}

\begin{document}

	\title{On Frequency Dependent Log-Optimal \\ Portfolio with Transaction Costs}
\date{}
	\author{Chung-Han Hsieh$^{\ast}$$\dag$\thanks{$^\ast \dag$Corresponding author: Chung-Han Hsieh.
			Email: \href{mailto: ch.hsieh@mx.nthu.edu}{ch.hsieh@mx.nthu.edu}. } and Yi-Shan Wong${\dag}$\thanks{\dag Email: \href{mailto: yishan.wong13@gmail.com}{yishan.wong13@gmail.com}. This paper was supported in part by the Ministry of Science and Technology, R.O.C. Taiwan, under Grants: MOST110--2222--E--007--005-- and MOST111--2221--E--007--124--.}\\
 Department of Quantitative Finance, \\National Tsing Hua University, Hsinchu 300044, Taiwan R.O.C.
	}
	
	\maketitle
	
\textbf{Abstract.}
		The aim of this paper is to investigate the impact of rebalancing frequency and transaction costs on the log-optimal portfolio, which is a portfolio that maximizes the \textit{expected logarithmic growth} rate of an investor's wealth.
		We prove that the frequency-dependent log-optimal portfolio problem with costs is equivalent to a concave program and provide a version of the dominance theorem with costs to determine when an  investor should invest all available funds in a particular asset. 
		Then, we show that transaction costs may cause a bankruptcy issue for the frequency-dependent log-optimal portfolio.
		To address this issue,  we approximate the problem to obtain a quadratic concave program and derive necessary and sufficient optimality conditions. 
		Additionally, we prove a version of the \textit{two-fund theorem}, which states that any convex combination of two optimal weights from the optimality conditions is still optimal.
		We test our proposed methods using both intraday and daily price data.
		Finally, we extend our empirical studies to an online trading scenario by implementing a sliding window approach. This approach enables us to solve a sequence of concave programs rather than a potentially computational complex stochastic dynamic programming problem.

	\vspace{5mm}

		\vspace{5mm}
\textit{Classcodes:} G11, G14, C61

	\section{Introduction}\label{section: introduction}
	The takeoff point of this paper is to study the celebrated log-optimal portfolio, which calls for maximizing the Expected Logarithmic Growth~(ELG) of an investor's wealth.
	This ELG maximization idea was introduced by~\cite{Kelly_1956} and is also known as the \textit{Kelly Criterion}.\footnote{See \cite{poundstone2010fortune} for a storytelling book that brought the Kelly criterion to the attention of practical investors.}   
	Some earlier works related to ELG maximization and its possible applications in gambling and trading can be found in \cite{breiman1961optimal,thorp1975portfolio, cover1984algorithm, algoet1988asymptotic, rotando1992kelly, browne1996portfolio}.
	Many subsequent papers contributed to the ELG problem and its various ramifications.
	For example, \cite{thorp2006kelly} studied ELG problems in blackjack, sports betting, and the stock market.
	\cite{maclean2010long, maclean2011does} summarized the good and bad properties of maximizing ELG.
	\cite{maclean2016optimal} studied the risky short-run properties of the ELG criterion.
	\cite{cover2006elements, luenberger2013investment} are textbooks that contain an introductory chapter on the ELG problem.
	\cite{kim2017comparison} demonstrated the superior return of the log-optimal portfolio compared to a traditional mean-variance portfolio in the Korean stock market.
	\cite{lo2018growth} connected ELG results to population genetics and discussed testable findings using experimental evolution.
	More recently, \cite{wu2020analysis} analyzed  Kelly betting in finite repeated games.
	\cite{maclean2022kelly} studied the ELG problem in a regime-switching market framework.
	\cite{wang2022data} proposed a data-driven log-optimal portfolio via a sliding window approach.
	However, among all of these papers, the effects of \textit{rebalancing frequency} have \textit{not} been extensively considered in previous~literature.

	\subsection{Rebalancing Frequency Considerations}
	There are some existing results regarding rebalancing frequency in a log-optimal portfolio, as found in~\cite{kuhn2010analysis, das2014computing, das2015computing}, and \cite{hsieh2021necessary}. 
	Specifically, \cite{kuhn2010analysis} considered a portfolio optimization with returns following a continuous geometric Brownian motion, but only focused on two extreme cases: High-frequency trading and buy and hold.
	On the other hand,~\cite{das2014computing} and \cite{das2015computing} studied log-optimal portfolio with the constant weight~$K$ selected without regard for the frequency.
	However, when the same weight~$K$ is used to find an optimal rebalancing period, the resulting ELG levels are arguably suboptimal. 
	Lastly, in our prior work \cite{hsieh2021necessary}, we formulated a discrete-time frequency-dependent log-optimal portfolio problem and derived various optimality conditions, but we did not consider the effects of the transaction costs.

	\subsection{Transaction Costs Considerations}
	Transaction costs are known to significantly impact  trading performance in practice; see \cite{cornuejols2006optimization, bogle2017little}. 
	These costs may include execution commissions, bid-ask spreads, latency costs, and the price impact of trading.
	As a result, an optimal trading policy may no longer be optimal when the transaction costs are not zero; see \cite{cvitanic2004introduction, hsieh2018frequency}.
	Consequently, many previous papers have focused on addressing the effect of transaction costs on a portfolio optimization problem. 
	
	One of the earliest models for \textit{expected} execution costs was developed by \cite{bertsimas1998optimal}, where dynamic programming techniques are applied. 
	\cite{almgren2001optimal} studied an optimal tradeoff between expected cost and risk.
	In the meantime, \cite{Magill1976, shreve1994optimal, muthuraman2006multidimensional} have all studied proportional transaction costs in continuous-time portfolio optimization problems, but none of these studies considered the effects of rebalancing frequency. 
	Other transaction cost models can be found in the literature; e.g., \cite{lobo2007portfolio} considered fixed transaction costs in a portfolio optimization problem. 
	\cite{woodside2013portfolio} studied fixed, and V-shaped variable transaction costs in a mean-variance model.
	A recent empirical study; see \cite{ruf2020impact}, analyzed portfolios' performance in the presence of proportional transaction costs under various discrete rebalancing frequencies, constituent list size, and renewing frequency. 
	While it compared trading performance under various classical portfolios arising in stochastic portfolio theory, it does not take portfolio optimization into account.

	Following the frequency-dependent portfolio optimization framework described by~\cite{hsieh2018rebalancing, hsieh2021necessary}, this paper extends the formulation to incorporate proportional transaction costs.
	When there are nonzero proportional transaction costs, the frequency-dependent log-optimal portfolio problem would be intractable since there may be no  (i.e., because bankruptcy might occur) or only a trivial solution (i.e., zero investments) for such a problem. 
	To address this issue,   we propose an approximation approach. 
	We also derive optimality conditions for the approximate problem.
	In addition, we prove a version of the Dominance Theorem involving proportional transaction costs, which shows under what circumstances a log-optimal investor would invest all available funds.

	\subsection{Contributions of the Paper}
	In Section~\ref{sec: problem formulation}, we formulate the frequency-dependent log-optimal portfolio problem involving proportional transaction costs. 
	We prove that the problem is equivalent to a concave program (see Lemma~\ref{lemma: ELG optimization as a concave program}) and show a version of the Dominance Lemma~\ref{lemma: dominance theorem with costs} with cost considerations. 
	We also state a sufficient condition to invest all available funds in a single asset.
	Then, in Section~\ref{section: Taylor}, we investigate bankruptcy issues when there are nonzero costs, see Lemma~\ref{lemma: survival trades under costs}. 
	Then by approximating the ordinary optimization problem with a concave quadratic program, we provide necessary and sufficient optimality conditions for an approximate log-optimal portfolio; see Lemma~\ref{lemma: Optimality conditions}. 
	We further prove a version of the Two-Fund Theorem~\ref{theorem: two-fund theorem}, demonstrating that a combination of two optimal weights is still optimal.
	Finally, in Section~\ref{section: online trading with sliding window approach}, we extend our theory to online trading via a sliding window approach.

	\section{Problem Formulation}\label{sec: problem formulation}	
	This section provides some background information and the formulation for the frequency-dependent log-optimal problem with costs.
	Let $N > 1$ be a terminal stage. 
	For stage~$k=0,1,\dots, N-1$, consider an investor forming a portfolio consisting of~$m \geq 2$ assets and assume that at least one asset is riskless with a rate of return $r_f \geq 0.$ That is, if an asset is riskless, its return is deterministic and is treated as a degenerate random variable with value $X(k) = r_f$ for all~$k$ with probability one.\footnote{In practice, the actual distribution of returns may not be available to the investor, but one can always estimate it and work with the \textit{empirical} surrogate.} 
	Alternatively, if Asset~$i$ is a risky asset whose price at time~$k$ is~$S_i(k)>0$, then its per-period return is given by
	$
	X_i(k) = \frac{S_i(k+1) - S_i(k)}{S_i(k)}.
	$
	In the sequel, for risky assets, we assume that the return vectors~$
	X(k):=\left[X_1(k) \, X_2(k)\,  \cdots \,X_m(k)\right]^T
	$
	have a known distribution and have components $X_i(\cdot)$ which can be arbitrarily correlated.\footnote{Again, if the $i$th asset is riskless, then we put $X_i(k) = r_f \geq 0$ with probability one. If an investor maintains \textit{cash} in its portfolio, then this corresponds to the case~{$r_f=0.$}} 
	We also assume that these vectors are i.i.d. with components satisfying
	$
	X_{\min,i}  \leq X_i(k) \leq X_{\max,i}
	$
	with known bounds above and with~$X_{\max,i}$ being finite and~{$X_{\min,i} > -1$}.
	The latter constraint on $X_{\min, i}$ means that the loss per time step is limited to less than~$100\%$ and the price of a stock cannot drop to zero.

	\subsection{Linear Policy and Unit Simplex Constraint}
	Consistent with the literature, e.g.,~\cite{barmish2015new,hsieh2019positive,hsieh2018frequency,hsieh2018rebalancing, zhang2001stock, primbs2007portfolio}, we consider a \textit{linear policy} with a weight vector $K \in \mathbb{R}^m$.
	Let $V(k)$ be the investor's account value at stage~$k$  and the weight for Asset $i$ is given by
	$
	0 \leq K_i \leq 1
	$
	represents the fraction of the account allocated to the~{$i$th} asset for~$i=1,\dots,m$. 
	Said another way,  the policy for the~$i$th asset is of a linear form
	$
	u_i(k) := K_iV(k).
	$
	Note that the number of shares invested on the~$i$th asset is $u_i(k)/S_i(k)$.
	Since~$K_i \geq 0$, the investor is going {\it long}.
	In view of this, and given that there is at least one riskless asset available, 
	we consider the \textit{unit simplex} constraint
	\begin{align} \label{eq: unit simplex}
		K \in {\mathcal K} := \left\{K \in \mathbb{R}^{m}: K_i \geq 0 \text{ for all $i = 1,\dots,m$}, \; \sum_{i=1}^m K_i = 1 \right\}
	\end{align}
	which is classical constraint in finance; e.g., see \cite{cvitanic2004introduction, cover2006elements, luenberger2013investment,cuchiero2019cover,hsieh2021necessary}.
	With~{$K \in \mathcal K$}, we guarantee that~100\% of the account is invested. 
	
	\begin{remark} \rm
		In the finance literature, it is known that  long-only constraints like (\ref{eq: unit simplex}) can be used to mitigate the  overconcentration of weight. These constraints can also assist in containing volatility and trading performance; see \cite{jagannathan2003risk}. 
	\end{remark}

	\subsection{Frequency Dependent Account Value Dynamics with Transaction Costs} 
	Letting $n \geq 1$ be the number of steps between rebalancings,  at time~{$k=0$}, the investor begins with initial investments~{$
		u(0) = \sum_{i=1}^m u_i(0)
		$}
	with~$u_i(0):= K_i V(0)$ with $K_i$ being the $i$th component of the portfolio weight $K$ satisfying constraint $(\ref{eq: unit simplex})$. 
	It is worth mentioning that the investment level~$u_i(0)$ can be converted to the number of shares  by dividing it by the price $S_i(0)$; i.e.,  $u_i(0)/S_i(0)$.
	The investor then waits~$n$ steps in the spirit of buy and hold. When~{$k=n$}, the investment control is updated to be~{$
		u(n) = \sum_{i=1}^m K_i V(n).
		$}
	Continuing in this manner, a waiting period of $n$ stages is enforced between each rebalance.

	To incorporate the transaction costs into the frequency-dependent framework, let $c_i \in [0, c_{\max}]$  be a \textit{percentage transaction costs} imposed on Asset~$i$ where $c_{\max} \in (0,1)$ is a predetermined maximum transaction cost.\footnote{Nowadays, while some online brokerage services offer fee-free trades  for certain exchange-traded funds~(ETFs) in the United States, transaction costs are typically required. 
		For example, trading on the Taiwan Stock Exchange typically incurs a transaction cost of $\alpha \cdot 0.1425\%$ of the trade value for some $\alpha \in (0,1)$. As a second example, using professional broker services such as Interactive Brokers Pro., may incur a fee of $\$0.005$ per share, with a minimum fee of \$1 dollar and a maximum fee of $1\%$ of the trade value. } 
	That is, at stage $k=0$, if one invests $u_i(0)$ at Asset~$i$, then the associated transaction costs in dollar is~$u_i(0) c_i \geq 0$.
	Then the dynamics of account value at stage~$n \geq 1$ is characterized by the following stochastic recursive equation:\footnote{At stage $\ell \geq 0$ and rebalancing period $n \geq 1$, the stochastic recursion of account value becomes
		$$
		V(n (\ell+1)) 
		= V(n \ell) + \sum_{i=1}^m \frac{u_i(n \ell)}{S_i(n \ell)} (S_i(n (\ell+1) ) - S_i(n \ell))- \sum_{i=1}^m u_i(n \ell)c_i.
		$$}
	\begin{align} \label{eq: account value in terms of stocks prices}
		V(n) 
		&= V(0) + \sum_{i=1}^m \frac{u_i(0)}{S_i(0)} (S_i(n) - S_i(0))- \sum_{i=1}^m u_i(0)c_i .
	\end{align}
	In the sequel, we may sometimes write $V_K(n)$ instead of~$V(n)$ to emphasize the dependence on portfolio weight $K$.

	\begin{remark}[Transaction Costs]\rm
		If there are no transaction costs; i.e., $c_i :=0$ for all~$i=1,2,\dots,m$, then the account value dynamics~(\ref{eq: account value in terms of stocks prices}) reduces to the existing formulation in \cite{rujeerapaiboon2018risk, hsieh2018rebalancing}; see also the frictionless market setting discussed in~\cite{merton1992continuous}.
	\end{remark}

	\subsection[Frequency-Dependent Optimization Problem]{Frequency-Dependent Optimization Problem} 
	Following previous research in~\cite{hsieh2018rebalancing, hsieh2021necessary}, to study the performance which is dependent on rebalancing frequency, for~$i=1, 2, \ldots, m$, we work with the  $n$-period \textit{compound~returns} for each asset $i$, call it~$\mathcal{X}_{n,i}$, defined as
	\[
	\mathcal{X}_{n,i} = \frac{S_i(n) - S_i(0)}{S_i(0)}.
	\]
	It is readily verified that
	$
	\mathcal{X}_{n,i} = \prod_{k=0}^{n-1} (1+X_i(k)) -1
	$
	and
	$
	-1 < \mathcal{X}_{\min,i} \leq \mathcal{X}_{n,i} \leq \mathcal{X}_{\max, i} 
	$
	where $\mathcal{X}_{\max,i} := (1+X_{\max,i})^n -1 $ and $\mathcal{X}_{\min,i}:= (1 + X_{\min, i})^n -1 >-1 $ for all $n\geq 1.$
	In the sequel, we work with the random vector~$\mathcal{X}_n$ having $i$th component~$\mathcal{X}_{n,i}$.  
	
	Now for any rebalancing period~$n \geq 1$, we define the expected logarithmic growth (ELG) 
	\begin{align*}
		{g_n}(K) 
		&:=  \frac{1}{n}\mathbb{E}\left[ \log \frac{V_K(n)}{V(0)} \right].
	\end{align*}
	Our goal is to solve the following frequency-dependent stochastic maximization problem:
	\begin{align} \label{problem: stochastic log-optimization problem}
		&\sup \; \{g_n(K): K \in \mathcal{K}\}  \\
		&{\rm s.t.}  \;\;\; V(n) = V(0) + \sum_{i=1}^m \frac{u_i(0)}{S_i(0)} (S_i(n) - S_i(0))- \sum_{i=1}^m u_i(0)c_i \nonumber 
	\end{align}
	where $\cal K$ is the unit simplex defined previously in Equation~(\ref{eq: unit simplex}).
	The following lemma shows that maximizing the frequency-dependent ELG with nonzero costs is indeed solving a concave program.
	
	\begin{lemma}[ELG Optimization as a Concave Program] \label{lemma: ELG optimization as a concave program}
		Fix $n \geq 1$ and $c_i \in (0,1)$.
		The frequency-dependent ELG optimization problem~(\ref{problem: stochastic log-optimization problem}) 
		is equivalent to
		\begin{align} \label{problem: simplified log-optimization problem}
			\max \left\{ \frac{1}{n}\mathbb{E}\left[ {\log (1 + {K^T}{ \widetilde{\mathcal{X}}_n})} \right]: K \in \mathcal{K} \right\}. 
		\end{align}
		where $\widetilde{\mathcal{X}}_n$ is a vector with the $i$th component given by
		$
		\widetilde{\mathcal{X}}_{n,i} := \mathcal{X}_{n,i} - c_i.
		$
		Additionally,  Problem~(\ref{problem: simplified log-optimization problem}) is a concave program.
	\end{lemma}	
	
	\begin{proof}
		We begin by observing that the account value dynamics 
		\begin{align}
			V(n) 
			&= V(0) + \sum_{i=1}^m \frac{u_i(0)}{S_i(0)} (S_i(n) - S_i(0))- \sum_{i=1}^m u_i(0)c_i \nonumber \\ 
			&= V(0) + \sum_{i=1}^m K_i V(0) \left( \frac{S_i(n) - S_i(0)}{S_i(0)} \right)- \sum_{i=1}^m u_i(0)c_i \nonumber \\ 
			&= V(0) + \sum_{i=1}^m K_i V(0) \mathcal{X}_{n,i}- \sum_{i=1}^m u_i(0)c_i \nonumber \\ 
			&= (1 + K^T\mathcal{X}_n) V(0) - \sum_{i=1}^m u_i(0)c_i \nonumber \\
			&= (1 + K^T\widetilde{\mathcal{X}}_n ) V(0) \label{eq: account value}
		\end{align}
		where $\widetilde{\mathcal{X}}_n$ is a vector with the $i$th component given by
		$
		\widetilde{\mathcal{X}}_{n,i} := \mathcal{X}_{n,i} - c_i
		$.
		Hence, it follows that
		\begin{align*}
			{g_n}(K) 
			&=  \frac{1}{n}\mathbb{E}\left[ \log \frac{V_K(n)}{V(0)} \right] 
			= \frac{1}{n}\mathbb{E}\left[ {\log (1 + {K^T}{ \widetilde{\mathcal{X}}_n})} \right].
		\end{align*}
		Therefore, the original Problem~(\ref{problem: stochastic log-optimization problem}) reduces to
		$
		\max \left\{g_n(K) = \frac{1}{n}\mathbb{E}\left[ {\log (1 + {K^T}{ \widetilde{\mathcal{X}}_n})} \right] : K \in \mathcal{K}\right\}.
		$
		The supremum operator is replaced by the maximum since $g_n(K)$ is continuous in $K$ over a compact domain $\cal K$.  Hence, the Weierstrass extremum theorem; see \cite{rudin1964principles}, guarantees that the maximum is attained.
		To complete the proof, it remains to show that Problem~(\ref{problem: simplified log-optimization problem}) is a concave program. This is accomplished by a standard convexity argument.
		Since $1+K^T\widetilde{ \mathcal{X} }_n$ is affine in $K$, taking the logarithm function yields a concave function. 
		Moreover, taking the expectation and multiplying a scaling factor $1/n$ preserve the concavity; see \cite{boyd2004convex}. 
		Therefore, the objective function $\frac{1}{n}\mathbb{E}\left[ {\log (1 + {K^T}{ \widetilde{\mathcal{X}}_n})} \right]$ is a concave function in $K$. On the other hand, $\cal K$ is a unit simplex which is a convex compact set. Therefore, the maximization considered in Problem~(\ref{problem: simplified log-optimization problem}) is a concave function over a convex compact set, hence, is a concave program.
	\end{proof}

	Henceforth, we denote  $g_n^*$ as the  optimal expected logarithmic~growth associated with the given rebalancing period of length $n$. 
	A vector~{$K^* \in \mathcal{K} \subset \mathbb{R}^m$} satisfying~$g_n(K^*) = g_n^*$ is called a~\textit{log-optimal weight}. 
	The portfolio that uses the log-optimal fraction vector is called \textit{frequency-dependent log-optimal portfolio}.

	\subsection{Dominance Lemma with Costs} \label{section: Dominance lemma with Costs}
	In this section, a version of the dominance lemma with costs is stated below.
	
	\begin{lemma}[Dominance] \label{lemma: dominance theorem with costs}
		Given a collection of $m\geq 2$ assets, if Asset $j$ satisfying
		$$
		\mathbb{E} \left[ \frac{ 1+\widetilde{\mathcal{X}}_{n,i}}{ 1+ \widetilde{\mathcal{X}}_{n,j}} \right] \leq 1,
		$$ for all $i \neq j$ with $i,j \in \{1,2,\dots,m\}$,
		then, for all $n\geq 1$, $g_n(K)$ is maximized by
		$
		K^* = e_j
		$
		where~$e_j$ is the unit vector in the $j$th coordinate direction. 
	\end{lemma}
	
	\begin{proof}
		To prove $K^* = e_j$, it suffices to show that $g_n(K) \leq g_n(e_j)$ for $K \in \mathcal{K}$. 
		For notational convenience, we work with the random vector $\widetilde{\mathcal{R}}_n := \widetilde{\mathcal{X}}_n + \textbf{1}$ where $\textbf{1} := [1 \; 1 \; \cdots \; 1]^T \in \mathbb{R}^m.$ 
		Since~$K^T\textbf{1} = 1$ for $K \in \mathcal{K}$, it follows that
		$
		g_n(K) = \frac{1}{n} \mathbb{E}[\log K^T\widetilde{\mathcal{R}}_n].
		$
		Hence, by applying Jensen's inequality to the concave logarithmic function, we obtain
		\begin{align*}
			g_n(K) - g_n(e_j) 
			& = \frac{1}{n} \mathbb{E} \left[ \log \frac{K^T \widetilde{\mathcal{R}}_n}{\widetilde{\mathcal{R}}_{n,j}}\right] \\
			& \leq \frac{1}{n} \log \mathbb{E}\left[ \frac{K^T \widetilde{\mathcal{R}}_n}{\widetilde{\mathcal{R}}_{n,j}} \right]\\
			&= \frac{1}{n} \log \left( \sum_{i=1}^m K_i \mathbb{E} \left[ \frac{ \widetilde{\mathcal{R}}_{n,i}}{\widetilde{\mathcal{R}}_{n,j}} \right] \right) \\
			&= \frac{1}{n} \log \left( \sum_{i=1}^m K_i \mathbb{E} \left[ \frac{ 1+\widetilde{\mathcal{X}}_{n,i}}{ 1+ \widetilde{\mathcal{X}}_{n,j}} \right] \right) \\
			&\leq \frac{1}{n} \log \left( \sum_{i=1}^m K_i \cdot 1 \right)\\
			& \leq \frac{1}{n} \log 1 = 0
		\end{align*}
		where the second last inequality holds since~$\mathbb{E} \left[ \frac{ 1+\widetilde{\mathcal{X}}_{n,i}}{ 1+ \widetilde{\mathcal{X}}_{n,j}} \right] \leq 1$ and the last inequality holds since~$\sum_{i=1}^mK_i = 1$. Therefore,
		$
		g_n(K)  \leq g_n(e_j). 
		$
	\end{proof}	
	
	\begin{remark} \rm
		Lemma~\ref{lemma: dominance theorem with costs} indicates that, under certain conditions,  an optimal log-optimal investor must invest all available funds in a specific asset when transaction costs are present.
		This result can be viewed as an extension of the Dominant Asset Theorem in \cite{hsieh2021necessary} to include transaction costs. 
		To see this, consider the case where there are no costs; i.e., $c_i = 0$ for all $i$, then $\widetilde{ \mathcal{X} }_{n,i} = \mathcal{X}_{n,i}$. This implies that the ratio
		\begin{align*}
			\mathbb{E} \left[ \frac{1+\widetilde{ \mathcal{X} }_{n,i}}{ 1 + \widetilde{\mathcal{X}}_{n,j}} \right] 
			& = \mathbb{E} \left[ \prod_{k=0}^{n-1} \frac{1+X_i(k) }{ 1+X_j(k)} \right] \\
			& = \left( \mathbb{E} \left[ \frac{1+X_i(0) }{ 1+X_j(0)} \right] \right)^n,
		\end{align*}
		where the last equality holds since $X_i(k)$ are i.i.d. in $k$.
		Thus, the condition $	\mathbb{E} \left[ \frac{ 1+\widetilde{\mathcal{X}}_{n,i}}{ 1+ \widetilde{\mathcal{X}}_{n,j}} \right] \leq 1
		$ reduces to a much simpler condition $\mathbb{E} \left[ \frac{1+X_i(0) }{ 1+X_j(0)} \right] \leq 1$, which is consistent with the Dominant Asset Theorem proved in \cite{hsieh2021necessary}.
	\end{remark}

	\section{An Approximate  Log-Optimal Portfolio Problem with Costs}\label{section: Taylor}
	When the transaction costs are present, the corresponding \textit{fee-adjusted return} is given by $ \widetilde{\mathcal{X}}_{n,i} = \mathcal{X}_{n,i} - c_i $.
	The next lemma provides a sufficient condition for ensuring that the trades survive up to stage $n$.

	\begin{lemma}[Probability of Having Survival Trades under Transaction Costs] \label{lemma: survival trades under costs}
		Fix $n \geq 1.$
		If ${X}_{\min,i} > c_i^{1/n}-1$ for all $i=1,2,\dots,m$, then the probability $P( V(n) > 0) =1.$
	\end{lemma}
	
	\begin{proof} 
		Let $n\geq 1$ be given.
		Observe that
		\begin{align}
			P(V(n) > 0)  
			&=  P( (1+K^T \widetilde{\mathcal{X}}_n)V(0)>0) \nonumber \\
			&=  P( 1+K^T \widetilde{\mathcal{X}}_n>0) \nonumber \\
			&=  P\left( \sum_{i=1}^m K_i( 1+  \widetilde{\mathcal{X}}_{n,i}) >0 \right)\nonumber \\
			&=  P\left( \sum_{i=1}^m K_i\left( 1+  \prod_{k=0}^{n-1} (1+ X_i(k))  - 1 - c_i \right) >0 \right)\nonumber \\
			&=  P\left( \sum_{i=1}^m K_i\left(   \prod_{k=0}^{n-1} (1+ X_i(k))  - c_i \right) >0 \right).  \label{eq: prob of survival trades}
		\end{align}
		where the third equality holds by invoking the fact that $\sum_{i=1}^m K_i = 1.$
		Now note that the event
		$
		\left\{ \sum_{i=1}^m K_i\left(   \prod_{k=0}^{n-1} (1+ X_i(k))  - c_i  \right) >0 \right\} \supseteq \left\{  \prod_{k=0}^{n-1} (1+ X_i(k))  >  c_i  \right\}.
		$
		With the aids of monotonicity of probability measure,  Equality~(\ref{eq: prob of survival trades}) becomes
		\begin{align*}
			P(V(n) > 0)  
			&\geq P\left( \prod_{k=0}^{n-1} (1+ X_i(k))  > c_i   \right).
		\end{align*}
		Since $\prod_{k=0}^{n-1} (1+X_i(k)) \geq (1+X_{\min,i})^n$ for all $i=1,2,\dots,m$ and ${X}_{\min,i} > c_i^{1/n}-1$ for all $i=1,2,\dots,m$, it follows that
		$
		\prod_{k=0}^{n-1} (1+X_i(k)) > c_i
		$ for all $i$.
		Therefore, we have
		$
		P(V(n)>0) = 1.
		$
	\end{proof}
	
	\begin{remark}\rm
		$(i)$  To assure a survival trade,  Lemma~\ref{lemma: survival trades under costs} indicates that the worst returns must be large enough. 
		Specifically, for $n =1$, it requires $X_{\min,i} > c_i -1$ for all $i$. On the other hand, if~$n \to \infty,$ which corresponds to buy and hold, then we must have $X_{\min, i} > 0$ for all $i$.  
		$(ii)$ On the converse of the Lemma~\ref{lemma: survival trades under costs}, it is readily verified that if $\min_i c_i>0$ and~$P(V(n)>0) = 1$ for~$n \geq 1$, then  $\sum_{i=1}^m K_i ((1+ \mu_i)^n - c_i) \geq 0$ where $\mu_i := \mathbb{E}[X_i(k)]$; see Lemma~\ref{lemma: second surivial lemma}. 
		This reveals a gap in obtaining a necessary condition for survival trades in Lemma~\ref{lemma: survival trades under costs}.
	\end{remark}

	Lemma~\ref{lemma: survival trades under costs} implies that for a fixed $c^*:=\max_i c_i \in (0,1)$, there exists $X_{\min,i} > -1$ such that~$V(n) \leq 0$ with positive probability.
	Said another way, the investor's account may experience a ``survival issue"  when the rebalancing frequency and costs are taken into  consideration. 
	In addition, this survival issue may cause the $g_n(K)$ to become ill-defined.
	To address this issue, we use a Taylor-based quadratic approximation of $g_n(K)$ around $K = \textbf{0}$; see~\cite{casella2021statistical} and write
	\begin{align} \label{eq: approximate expected log growth}
		g_n(K)
		&\approx \frac{1}{n} \left( K^T \mathbb{E}\left[ \widetilde{\mathcal{X}}_n \right] 
		- \frac{1}{2}K^T\mathbb{E}\left[ \widetilde{\mathcal{X}}_n \widetilde{\mathcal{X}}_n^T \right] K \right):=\widehat{g}_n(K).
	\end{align}
	
	It is well-known that such a quadratic approximation is accurate for small returns; see \cite{pulley1983mean}.\footnote{Without loss of generality, set $c_i:=0$ for all $i$. Then the Taylor expansion of $\mathbb{E}\left[ \log(1+K^T\mathcal{X}_n) \right]  = \mathbb{E}\left[ \sum_{d=1}^\infty (-1)^{d+1} \frac{ (K^T\mathcal{X}_n)^d}{d} \right]$ converges for all $K \in \mathcal{K}$ if $|K^T \mathcal{X}_n| \leq 1 $ with probability one. 
	}
	Hence, in the sequel, we consider an \textit{approximate} frequency-dependent log-optimal portfolio problem with costs as follows:
	\begin{align}
		\max  \left\{\widehat{g}_n(K): K \in \mathcal{K} \right\}. \label{problem: approximate log-optimal portfolio problem}
	\end{align}
	
	\begin{remark} \rm
		It is readily verified that	the approximate problem~(\ref{problem: approximate log-optimal portfolio problem}) described above is a \textit{concave quadratic program}, which enables us to solve it in an efficient manner; e.g., see \cite{diamond2016cvxpy}.
	\end{remark}
	
	\subsection{Optimality Conditions}
	In this section, we investigate the optimality conditions for the approximate frequency-dependent log-optimal problem~(\ref{problem: approximate log-optimal portfolio problem}).
	
	\begin{lemmarep}[Necessity and Sufficiency] \label{lemma: Optimality conditions}
		Fix $n \geq 1.$
		Given a percentage costs $c_i \in (0,1)$ for~$i=1,2,\ldots,m$, the portfolio weight $\widehat{K}^{*} \in \mathcal{K}$ is optimal to the approximate frequency-dependent log-optimal problem~(\ref{problem: approximate log-optimal portfolio problem}), if and only if 
		\begin{align} \label{eq:condition 1}
			\mathbb{E}\left[\widetilde{\mathcal{X}}_{n,i}\right] 
			-\sum_{j=1}^m \widehat{K}_j^* \mathbb{E}\left[\widetilde{\mathcal{X}}_{n,i} \widetilde{\mathcal{X}}_{n,j} \right] 
			&= \widehat{K}^{*T} \mathbb{E}\left[ \widetilde{\mathcal{X}}_{n} \right]
			- \widehat{K}^{*T} \mathbb{E}\left[ \widetilde{\mathcal{X}}_n \widetilde{\mathcal{X}}_n^T \right] \widehat{K}^*,
			\text{ if $\widehat{K}_i^* > 0$ }\\
			\mathbb{E}\left[\widetilde{\mathcal{X}}_{n,i}\right] 
			-\sum_{j=1}^m \widehat{K}_j^* \mathbb{E}\left[\widetilde{\mathcal{X}}_{n,i} \widetilde{\mathcal{X}}_{n,j} \right] 
			&\leq \widehat{K}^{*T} \mathbb{E}\left[ \widetilde{\mathcal{X}}_{n} \right]
			- \widehat{K}^{*T} \mathbb{E}\left[ \widetilde{\mathcal{X}}_n \widetilde{\mathcal{X}}_n^T \right] \widehat{K}^*,
			\text{ if $\widehat{K}_i^* = 0$ }\label{eq:condition 2}
		\end{align}
	\end{lemmarep}
	
	\begin{proof} 
		Let $n \geq 1$ and $c_i \in (0,1)$ for all $i$ be given.
		We begin by considering an equivalent constrained stochastic minimization problem described as follows:
		\begin{align*}
			&\min_{K} -K^T \mathbb{E}\left[ \widetilde{\mathcal{X}}_n \right] 
			+ \frac{1}{2}K^T\mathbb{E}\left[ \widetilde{\mathcal{X}}_n \widetilde{\mathcal{X}}_n^T \right] K\\
			&{\rm s. t.} \; K^T \mathbf{1} - 1 = 0;\\
			&\;\;\;\;\; -K^T e_i \leq 0, \; i=1,2, \ldots,m
		\end{align*}
		where $e_i \in \mathbb{R}^m$ is unit vector having one at the $i$th component and zeros on the other components. 
		Consider the Lagrangian 
		\begin{align*}
			\mathcal{L}(K, \lambda, \mu) &:= -K^T \mathbb{E}\left[ \widetilde{\mathcal{X}}_n \right]
			+ \frac{1}{2}K^T\mathbb{E}\left[ \widetilde{\mathcal{X}}_n \widetilde{\mathcal{X}}_n^T  \right]K +\lambda(K^T \mathbf{1}-1) 
			- \mu^{T}K.
		\end{align*}
		By the Karush-Kuhn-Tucker (KKT) conditions; e.g., see~\cite[Chapter~5]{boyd2004convex}, if $\widehat{K}^{*}$ is a local maximum then there is a scalar $\lambda \in \mathbb{R}^1$ and a vector $\mu \in \mathbb{R}^m$ with component~$\mu_j \geq 0$ such that, for $i=1,2, \ldots,m$,
		\begin{align}
			&-\mathbb{E}\left[\widetilde{\mathcal{X}}_{n,i} \right] + \sum_{j=1}^m \widehat{K}_j^* \mathbb{E}\left[ \widetilde{\mathcal{X}}_{n,i} \widetilde{\mathcal{X}}_{n,j} \right] + \lambda - \mu_i = 0  \label{eq: partial Ki} \\
			& \widehat{K}^{*T} \mathbf{1} - 1 = 0  \label{eq: partial lambda} \\
			& \mu_i \widehat{K}_i^* = 0 .    
			\label{eq: Mu_i x Ki* = 0} 
		\end{align}
		From Equation~(\ref{eq: partial Ki}), we obtain, for $i=1, \ldots,m$,
		\begin{equation} \label{eq: Mu_i with lambda}
			\mu_i = -\mathbb{E}\left[ \widetilde{\mathcal{X}}_{n,i} \right] + \sum_{j=1}^m \widehat{K}_j^* \mathbb{E}\left[ \widetilde{\mathcal{X}}_{n,i} \widetilde{\mathcal{X}}_{n,j} \right] + \lambda . 
		\end{equation}
		Since $\mu_i \widehat{K}_i^* = 0$ for all $i$, we take weighted sum of Equation~(\ref{eq: Mu_i with lambda}); i.e.,
		\begin{equation} \label{eq: weighted sum of mu_i Ki}
			\sum_{i=1}^m \mu_i \widehat{K}_i^* = -\widehat{K}^{*T} \mathbb{E}\left[ \widetilde{\mathcal{X}}_{n} \right] + \widehat{K}^{*T} \mathbb{E}\left[ \widetilde{\mathcal{X}}_n \widetilde{\mathcal{X}}_n^T \right] \widehat{K}^* + \lambda = 0.
		\end{equation}
		This implies that
		$    \lambda = \widehat{K}^{*T}\mathbb{E}\left[ \widetilde{\mathcal{X}}_n \right] - \widehat{K}^{*T} \mathbb{E}\left[ \widetilde{\mathcal{X}}_n \widetilde{\mathcal{X}}_n^T \right] \widehat{K}^*.
		$
		Substituting this into Equation~(\ref{eq: Mu_i with lambda}), we have, for $i=1,2,\ldots,m$,
		\begin{align} \label{eq: mu_i}
			\mu_i = -\mathbb{E}\left[ \widetilde{\mathcal{X}}_{n,i} \right] + \sum_{j=1}^m \widehat{K}_j^* \mathbb{E}\left[ \widetilde{\mathcal{X}}_{n,i} \widetilde{\mathcal{X}}_{n,j} \right] 
			+ \widehat{K}^{*T} \mathbb{E}\left[ \widetilde{\mathcal{X}}_n \right] - \widehat{K}^{*T} \mathbb{E}\left[ \widetilde{\mathcal{X}}_n \widetilde{\mathcal{X}}_n^T \right] \widehat{K}^*.
		\end{align}
		From Equation~(\ref{eq: mu_i}) and the fact that~$\mu_i \widehat{K}_i^* = 0$, it follows that for
		$i=1,2,\ldots,m$, if $\widehat{K}_i^* > 0$, then~$\mu_i = 0 $ and 
		$$  \mathbb{E}\left[\widetilde{\mathcal{X}}_{n,i}\right] -\sum_{j=1}^m \widehat{K}_j^* \mathbb{E}\left[\widetilde{\mathcal{X}}_{n,i} \widetilde{\mathcal{X}}_{n,j} \right] \nonumber 
		= \widehat{K}^{*T} \mathbb{E}\left[ \widetilde{\mathcal{X}}_{n} \right]
		- \widehat{K}^{*T} \mathbb{E}\left[ \widetilde{\mathcal{X}}_n \widetilde{\mathcal{X}}_n^T \right] \widehat{K}^*.
		$$
		On the other hand,
		if $\widehat{K}_i^* = 0$, then $\mu_i \geq 0$ and
		$$    \mathbb{E}\left[\widetilde{\mathcal{X}}_{n,i}\right] -\sum_{j=1}^m \widehat{K}_j^* \mathbb{E}\left[\widetilde{\mathcal{X}}_{n,i} \widetilde{\mathcal{X}}_{n,j} \right] \nonumber 
		\leq \widehat{K}^{*T} \mathbb{E}\left[ \widetilde{\mathcal{X}}_{n} \right]
		- \widehat{K}^{*T} \mathbb{E}\left[ \widetilde{\mathcal{X}}_n \widetilde{\mathcal{X}}_n^T \right] \widehat{K}^*.
		$$
		
		To prove sufficiency, let $\widehat{K}^* \in \mathcal{K}$ and satisfies the conditions~(\ref{eq:condition 1}) and~(\ref{eq:condition 2}). 
		Then it follows that there exists~$\lambda \in \mathbb{R}$ and $\mu_j > 0$ such that the KKT conditions~(\ref{eq: partial Ki}) to (\ref{eq: Mu_i x Ki* = 0}) hold at~$\widehat{K}^*$. 
		Since the constrained minimization problem is a convex optimization problem, it follows that the KKT conditions are also sufficient for optimality. Hence, $\widehat{K}^*$ is optimal; see~\cite{ boyd2017multi}. 
	\end{proof}
	
	\begin{remark} \rm \label{rmk:approx K* is close to true K*}
		Let $\widehat{K}^*$ be the optimum obtained by solving the approximate frequency-dependent log-optimal portfolio problem~(\ref{problem: approximate log-optimal portfolio problem}) and $K^*$ be the true log-optimum. Using Jensen's inequality, we~have 
		\begin{align*}
			0 \leq g(K^*) - g(\widehat{K}^*) 
			&=  \mathbb{E} \left[ \log \frac{1+K^{*^T} \widetilde{ \mathcal{X} }_n }{1 + \widehat{K}^{*^T} \widetilde{ \mathcal{X} }_n }\right]\\
			& \leq \log \mathbb{E} \left[  \frac{1+K^{*^T} \widetilde{ \mathcal{X} }_n }{1 + \widehat{K}^{*^T} \widetilde{ \mathcal{X} }_n }\right].
		\end{align*}
		The right-hand side is approximately zero when $K^* \approx \widehat{K}^*$. As we will see later in this paper, this is typically the case. 
		More interestingly, Lemma~\ref{lemma: Optimality conditions} serves to compliment Lemma~\ref{lemma: dominance theorem with costs} by characterizing the log-optimal weights; see Example~\ref{example: two asset toy example} below. 
	\end{remark}

	\begin{example}[Two-Asset Toy Example] \label{example: two asset toy example} \rm  
		To demonstrate the application of Lemmas~\ref{lemma: dominance theorem with costs} and~\ref{lemma: Optimality conditions}, we first consider a high-frequency investor who rebalances her portfolio at every period; i.e., $n:=1$. 
		Specifically,  consider a two-asset portfolio including a risk-free cash asset with zero interest rate; i.e. $X_1(k) := r_f = 0$ with probability one and a risky asset with a binomial return $X_2(k) \in  \{ -\frac{1}{2}, \frac{1}{2} \}$ with probability 
		$
		P\left( X_2(k) =\frac{1}{2} \right) := p \in \left( \frac{1}{2} + c_2, 1 \right).
		$
		The transaction costs are $c_1 = 0$ for cash and~$c_2 < 1/2$ for the risky asset. 
		If $\widehat{K}_2^* > 0$, by Lemma~\ref{lemma: Optimality conditions}, we have
		\begin{align*}
			& (1-\widehat{K}_2^*)\left( p-\frac{1}{2}-c_2 \right)  - \widehat{K}_2^* (1 - \widehat{K}_2^*) \left[\frac{1}{4}-2c_2 \left( p-\frac{1}{2} + c_2^2 \right) \right] = 0.
		\end{align*}
		This implies that $\widehat{K}_2^* = \frac{-(4c_2-4p+2)}{4c_2^2+4c_2-8c_2p+1}$.
		Incorporating with Lemma~\ref{lemma: dominance theorem with costs}, we conclude
		\begin{equation}
			K_2^* := 
			\begin{cases}
				\widehat{K}_2^*  & \text{if} \; p\in \left(\frac{1}{2}+c_2,\frac{4c_2^2+8c_2+3}{4+8c_2}\right] \\
				1 & \text{if} \; p\in \left[\frac{4c_2^2+8c_2+3}{4+8c_2}, 1\right)
			\end{cases}
		\end{equation}
		and $K_1^*= 1-K_2^*.$
		Note that if $c_2 = 0$, then $K_2^* = 2(2p-1)$ for $p \in \left(\frac{1}{2},  \frac{3}{4}\right]$ or $K_2^* = 1$ for $p \in \left[\frac{3}{4}, 1\right]$, which reduces to the classical ELG result in gambling; see \cite{Kelly_1956, hsieh2018frequency}.
		
		To see the effect of rebalancing period $n > 1$, we consider a second example with $n=2$; i.e., one rebalances the portfolio for every two periods.
		For $c_2 \in \left[0, \frac{1}{4}\right)$, applying Lemmas~\ref{lemma: dominance theorem with costs} and~\ref{lemma: Optimality conditions} yield
		\begin{equation}
			K_2^* := 
			\begin{cases}
				\widehat{K}_2^*, & \text{if} \; p\in \left( \frac{1}{2}+c_2, -\frac{4c_2-9}{8c_2+6}-\frac{1}{4}C\right)\\
				1, & \text{if} \; p\in \left[ -\frac{4c_2-9}{8c_2+6}-\frac{1}{4}C, 1\right],
			\end{cases}
		\end{equation}
		where 
		$
		\widehat{K}_2^* = \frac{16p^2+16p-16c_2-12}{16c_2^2+24c_2+32p^2-16p-32p^2c_2-32pc_2+9},
		$
		and
		$
		C := \frac{\sqrt{-256c_2^4+384c_2^3+640c_2^2-504c_2+81}}{4c_2+3}
		$ and $K_1^* = 1-K_2^*$.
		If $c_2=0$, we have $K_2^*=\frac{16p^2+16p-12}{32p^2-16p+9}$ for $ p \in \left(\frac{1}{2}, \frac{3}{4}\right)$ and $K_2^*=1$ for $p \in \left[\frac{3}{4}, 1\right]$. 
	\end{example}

	\section{Feasible Region and Efficient Frontier}
	Similar to how the performance of a portfolio can be characterized by its expected return and variance in the celebrated Markowitz framework, the performance of log-optimal portfolios can be characterized by the \textit{expected logarithmic growth} and \textit{variance of the logarithmic growth} and plotted on a two-dimensional diagram; see \cite{luenberger2013investment}.
	The region mapped out by all possible portfolios defines the \textit{feasible region}. 
	That is, for any fixed $n \geq 1$,  we consider
	\[
	K \mapsto \left( \mathbb{E}\left[ \log \frac{V_K(n)}{V(0)} \right], {\rm var}\left(\log \frac{V_K(n)}{V(0)} \right) \right) \subset \mathbb{R}^2.
	\]
	As demonstrated later in Example~\ref{example: two fund thm for five asset portfolio}, the feasible region is \textit{convex to the left}. This means that if we take any two points within the region, the straight line connecting them does not cross the left boundary of the feasible region.
	A similar idea about analyzing the efficient frontier analytically can be found in \cite{merton1972analytic}.

	\subsection{A Version of The Two-Fund Theorem}
	In the approximate log-optimal portfolio problem, as defined in~(\ref{problem: approximate log-optimal portfolio problem}), the upper left-hand portion at the boundary of the feasible region is referred to as the \textit{approximate efficient frontier}.
	This frontier is considered \textit{efficient} in terms of expected logarithmic growth rate and its variance; see also \cite[Chapter 14]{luenberger2013investment}.
	Then, with the aid of Lemma~\ref{lemma: Optimality conditions}, we can obtain a version of the two-fund theorem, which states that any convex combination of two optimal weights from the optimality conditions is still optimal.

	\begin{theoremrep}[A Version of Two-Fund Theorem] \label{theorem: two-fund theorem}
		Let $K', K'' \in \mathcal{K}$ be two weights satisfying the optimality conditions stated in Lemma~\ref{lemma: Optimality conditions}. Define a convex combination $\overline{K}_\alpha := \alpha K' + (1-\alpha)K''$ with $\alpha \in [0,1]$.  
		Then $\overline{K}_\alpha$ also satisfies the optimality conditions.
	\end{theoremrep}
	
	\begin{proof} 
		Take $K'$ and $K''$ be two weights satisfying Equations~(\ref{eq: partial Ki}) to (\ref{eq: Mu_i x Ki* = 0}),  for all $\alpha \in [0,1]$, we must show that the convex combination of the two weights $K'$ and $K''$,
		$
		\overline{K}_\alpha := \alpha K' + (1-\alpha)K''
		$,  with the $j$th component $\overline{K}_{\alpha,j}$,
		also satisfies the same optimality equations. 
		In particular, we begin by proving that $\overline{K}_\alpha$ satisfies Equation~(\ref{eq: partial lambda}). 
		Indeed, we observe that
		\begin{align}
			(\alpha K' + (1-\alpha)K'')^T \textbf{1} -1 &=	\alpha K'^T\textbf{1}  + (1-\alpha)K''^T \textbf{1} -1 \label{eq: linear combination of weight}
		\end{align}
		where $\textbf{1} := [1 \; 1 \; \cdots \; 1]^T \in \mathbb{R}^m.$
		Since $K',K''$ satisfy Equation~(\ref{eq: partial lambda}), it follows that $ K'^T \textbf{1} = 1$ and~$K''^T \textbf{1} = 1.$ 
		Therefore, Equation~(\ref{eq: linear combination of weight}) becomes
		$	(\alpha K' + (1-\alpha)K'')^T \textbf{1} -1  = \alpha + (1-\alpha) = 1
		$	
		which proves that the convex combination $\overline{K}_\alpha$ satisfies Equation~(\ref{eq: partial lambda}).
		To see it also satisfies $\mu_i \widehat{K}_i^* = 0$ for $i=1,\dots, m$, we observe that
		\begin{align*}
			\mu_i (\alpha K_i' + (1-\alpha)K_i'') 
			&= \alpha \mu_i K_i' + (1-\alpha) \mu_i K_i'' \\
			&= \alpha \cdot 0 + (1-\alpha) \cdot 0 = 0. 
		\end{align*}
		To complete the proof, we show that $\overline{K}_\alpha$ satisfies Equation~(\ref{eq: partial Ki}). It suffices to show that for~$i=1,\dots,m$,
		$
		-\mathbb{E}\left[\widetilde{\mathcal{X}}_{n,i} \right] + \sum_{j=1}^m \overline{K}_{\alpha,j} \mathbb{E}\left[ \widetilde{\mathcal{X}}_{n,i} \widetilde{\mathcal{X}}_{n,j} \right] + \lambda = \mu_i. 
		$
		Note that the left-hand side using $\overline{K}_\alpha$ yields
		\begin{align*} 
			&-(\alpha + (1-\alpha))\mathbb{E}\left[ \widetilde{\mathcal{X}}_{n,i} \right] 
			+ \sum_{j=1}^m (\alpha K_j' + (1-\alpha) K_j'') \mathbb{E}\left[ \widetilde{\mathcal{X}}_{n,i} \widetilde{\mathcal{X}}_{n,j} \right] + (\alpha + (1-\alpha)) \lambda  \\
			&= \alpha \left( - \mathbb{E}\left[ \widetilde{\mathcal{X}}_{n,i} \right]  + \sum_{j=1}^m  K_j'  \mathbb{E}\left[ \widetilde{\mathcal{X}}_{n,i} \widetilde{\mathcal{X}}_{n,j} \right] + \lambda \right)  + (1-\alpha )\left(-\mathbb{E}\left[ \widetilde{\mathcal{X}}_{n,i} \right]  + \sum_{j=1}^m  K_j'' \mathbb{E}\left[ \widetilde{\mathcal{X}}_{n,i} \widetilde{\mathcal{X}}_{n,j} \right] +  \lambda \right) \\
			&= \alpha \mu_i + (1-\alpha) \mu_i = \mu_i
		\end{align*}
		which completes the proof.
	\end{proof}

	\begin{example}[Five-Asset Portfolio with Intraday Minute-by-Minute Data] \rm \label{example: two fund thm for five asset portfolio}
		This example illustrates the feasible region, efficient frontier, and Two-Fund Theorem~\ref{theorem: two-fund theorem} using a five-asset portfolio consisting of  a bank account, Vanguard Total Stock Market Index Fund ETF (Ticker: VTI), Vanguard Total Bond Market Index Fund ETF~(Ticker: BND), Vanguard Emerging Markets Stock Index Fund ETF (Ticker:~VWO), and Bitcoin to the USD exchange rate~(Ticker: XBTUSD). 
		The portfolio is well-diversified, covering the large US-Euro stock market, the global bond market, and cryptocurrency. 
		Here, transaction costs~$c_i = 0.001\%$ are imposed on the ETFs~(i.e.,~$i \in \{ {\rm VTI, BND, VWO}\}$) and costs $c_{\rm XBTUSD} = 0.1\%$ on the XBTUSD.\footnote{According to the platform Binance \href{binance.com/en}{binance.com/en}, regular users are charged a transaction cost of $0.1\%$ for Bitcoin trades.} 
		Besides, investors receive interest at a (per-minute) rate~$r_f = 0.0001\%$ if they keep their funds in the bank account.
		The data used in this example spans from~$09:30:00$~AM to~$15:59:00$~PM on December~3,~2021, where the associated price trajectories for the four risky assets are shown in Figure~\ref{fig: stock price}.\footnote{The price data for the four underlying risky assets (VTI, BND, VWO, XBTUSD) are retrieved from the Bloomberg terminal (accessed on November 17, 2022).} 
		To derive the approximate log-optimal portfolio and examine its trading performance, we split the entire data set into two parts: The first portion from $09:30:00$~AM to~$12:29:00$~PM is for the in-sample optimization,  and the second portion~$12:30:00$~PM to~$15:59:00$~PM is for the out-of-sample~testing.\footnote{
			This will be demonstrated later in Example~\ref{example: Trading Performance with Different Rebalancing Periods} in the next section.}

		\begin{figure}[h!]
			\centering
			\includegraphics[width=.9\linewidth]{"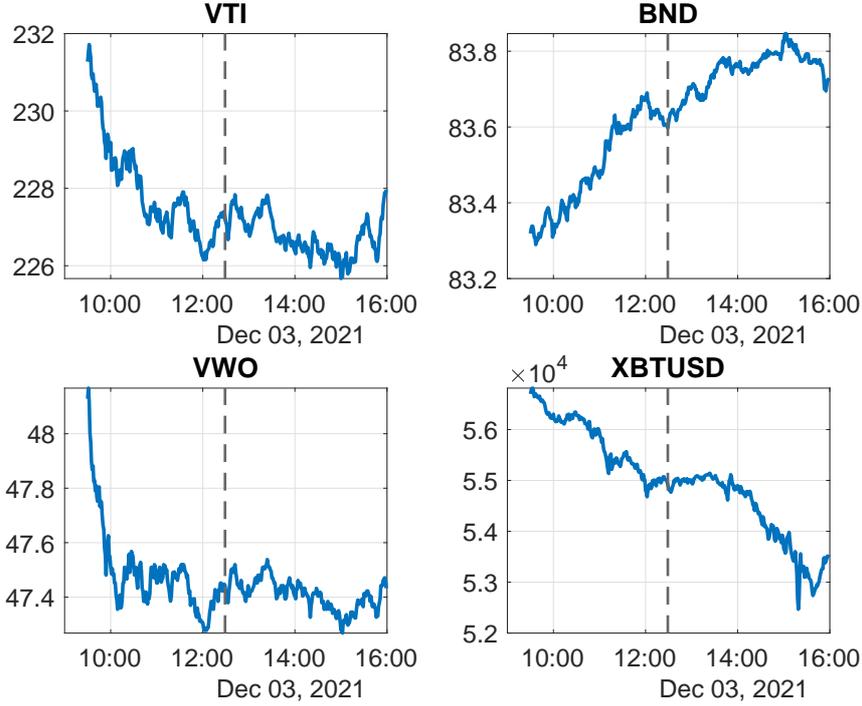"}
			\caption{Intraday Minute-by-Minute Prices for VTI, BND, VWO, and XBTUSD.}
			\label{fig: stock price}
		\end{figure}

		Fix $n \geq 1$. We define the approximate feasible region~$\mathcal{H} :=\left\{\left(  \widehat{g}_n(K), \; {\rm var}\left( \log \frac{V_K(n)}{V(0)} \right) \right): K \in \mathcal{K} \right\}$.
		Figures~\ref{fig: two funds(n=1)} and~\ref{fig: two funds(n=5)} show the points in $\mathcal{H}$ and the approximate efficient frontier for different rebalancing periods $n=1$ and $n=5$, respectively. 
		As predicted by Theorem~\ref{theorem: two-fund theorem},  any convex combination of two optimal weights $K'$ and $K''$ satisfying optimality conditions~\ref{lemma: Optimality conditions}, denoted as~$\overline{K}_\alpha = \alpha K' + (1-\alpha) K''$ with~$\alpha \in [0,1]$, satisfies the optimality conditions.
		Interestingly, it also lies on the approximate efficient frontier due to the small scale of the minute-by-minute price data\footnote{This phenomenon disappears when  using daily data; see also Remark~\ref{remark: comments on two-fund and efficient frontier} for more information.}; see Figures~\ref{fig: two funds(n=1)} and~\ref{fig: two funds(n=5)} for an example with $\alpha = 0.5$.
		Similar findings also hold for other rebalancing periods $n > 5$. 
	\end{example}
	
	\begin{remark}\rm \label{remark: comments on two-fund and efficient frontier}
		While not pursued further in this paper, the optimality conditions derived in Lemma~\ref{lemma: Optimality conditions} only consider the approximate logarithmic growth function $\widehat{g}_n(K)$ without taking into account the log-variance ${\rm var}(\log V_n(K)/V(0))$. As a result, to ensure that any convex combination of two points on the approximate efficient frontier is still on the frontier, the log-variance must be included in the optimization problem~(\ref{problem: approximate log-optimal portfolio problem}). This topic presents a promising research direction.
	\end{remark}

	\begin{figure}[h!]
		\centering
		\includegraphics[width=.8\linewidth]{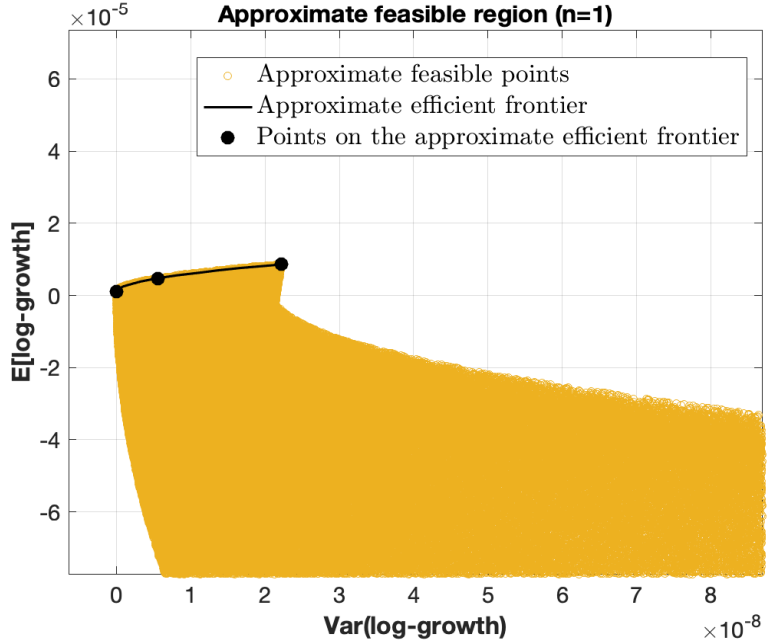}
		\caption{An illustration of Feasible Set, Efficient Frontier, and Two-Fund Theorem ($\overline{K}_\alpha$ with~$\alpha=0.5$) using Rebalancing Period $n=1$ (Minute). } 
		\label{fig: two funds(n=1)}
	\end{figure}
	
	\begin{figure}[h!]
		\centering
		\includegraphics[width=.8\linewidth]{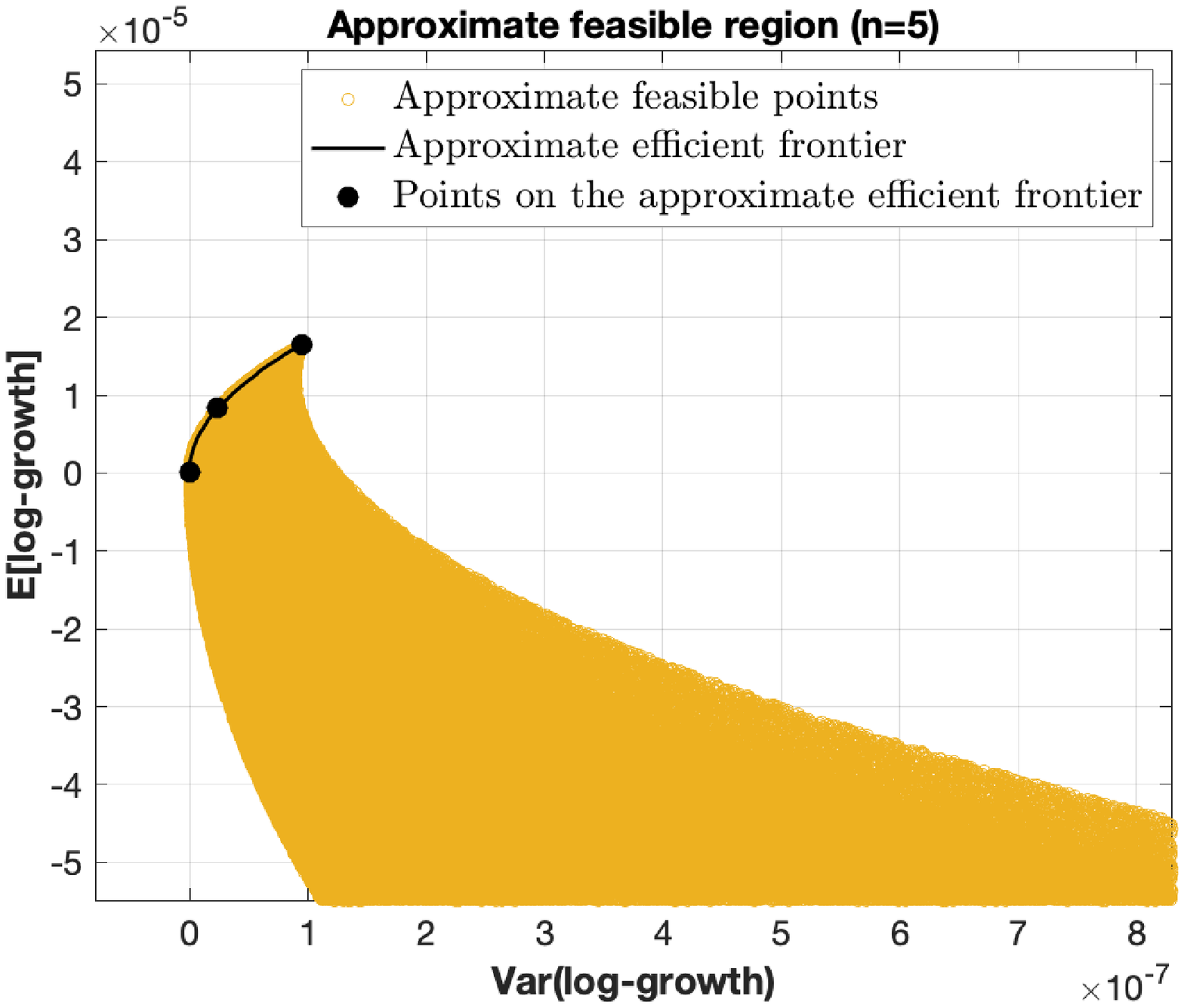}
		\caption{An Illustration of Feasible Set, Efficient Frontier, and Two-Fund Theorem ($\overline{K}_\alpha$ with~$\alpha=0.5$) using Rebalancing Period~$n=5$ (Minutes).} 
		\label{fig: two funds(n=5)}
	\end{figure}

	\section{Illustrative Examples} \label{section: illustrative examples}
	This section presents empirical examples to demonstrate the validity of our theory. 
	In the first two examples, we use the same intraday data set as Example~\ref{example: two fund thm for five asset portfolio} to compare the log-optimal and approximate log-optimal results. 
	We evaluate the impact of different rebalancing periods and levels of costs on trading performance.
	The third example examines the capability of our theory to handle the mid-sized portfolio case by considering a portfolio of thirty-two assets (with a Bank account, Dow-30 stocks, and cryptocurrency) using daily historical price data. 
	
	\begin{example}[Five-Asset Portfolio Revisited] \rm \label{example: The Similarity of K* and KHatStar}
		This example demonstrates that the approximate optimal weights~$\widehat{K}^*$ from Lemma~\ref{lemma: Optimality conditions} is sufficiently close to the optimal weights $K^*$. 
		To demonstrate this, we choose the weights $K^*$ on the efficient frontier that satisfy the logarithmic variance condition:
		$
		{\rm var}\left( \log \frac{V_{K^*}(n)}{V(0)} \right)  \equiv {\rm var}\left( \log \frac{V_{\widehat{K}^*}(n)}{V(0)} \right).
		$
		Figures~\ref{fig: weights (n=1)} and \ref{fig: weights (n=5)} show the portfolio weights of the two trading strategies: the approximate log-optimal weights $\widehat{K}^*$, and the true log-optimal weights~$K^*$ with different rebalancing periods~$n=1$ and~$n=5$. 
		The results show that the weights of the two strategies are nearly identical, i.e.,~$\widehat{K}_i^* \approx K_i^*$, for all~$i = 1, 2, \ldots, 5$. This suggests that the approximate optimal weights $\widehat{K}^*$ are a good approximation of the true optimal weights $K^*$.
		While not showing here, it is also worth mentioning that if the transaction costs are sufficiently large,  then both of the optima $K^*$ and $\widehat{K}^*$ will tend to fully invest in the bank account, meaning that~$K_{\rm Bank\; account}^* \approx \widehat{K}_{\rm Bank \; account}^* \approx 1.$

		\begin{figure}[h!]
			\centering
			\includegraphics[width=.8\linewidth]{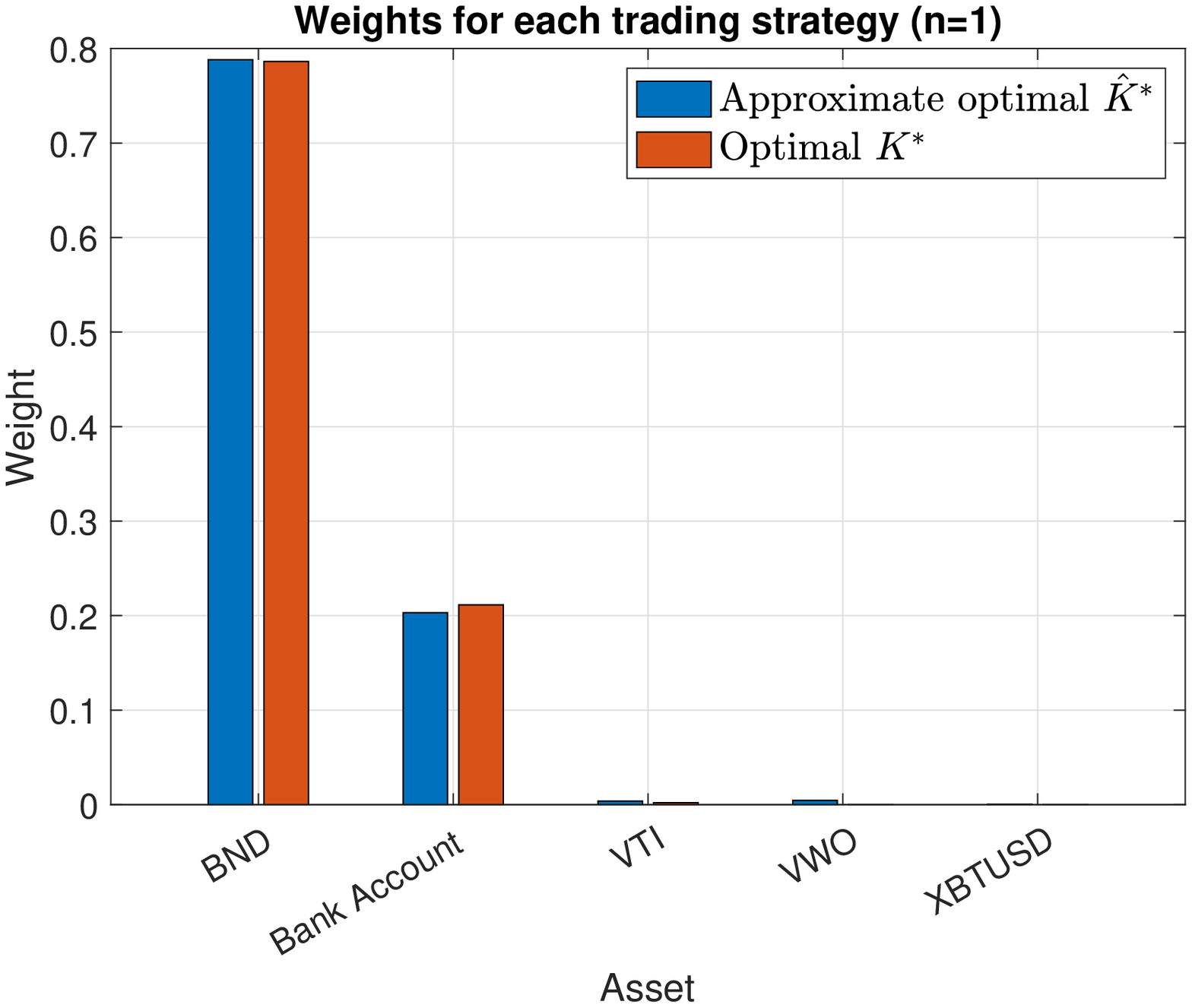}
			\caption{Portfolio Weights $K^*$ versus $\widehat{K}^*$ with Rebalancing Period $n=1$ (Minute).}
			\label{fig: weights (n=1)}
		\end{figure}
		
		\begin{figure}[h!]
			\centering
			\includegraphics[width=.8\linewidth]{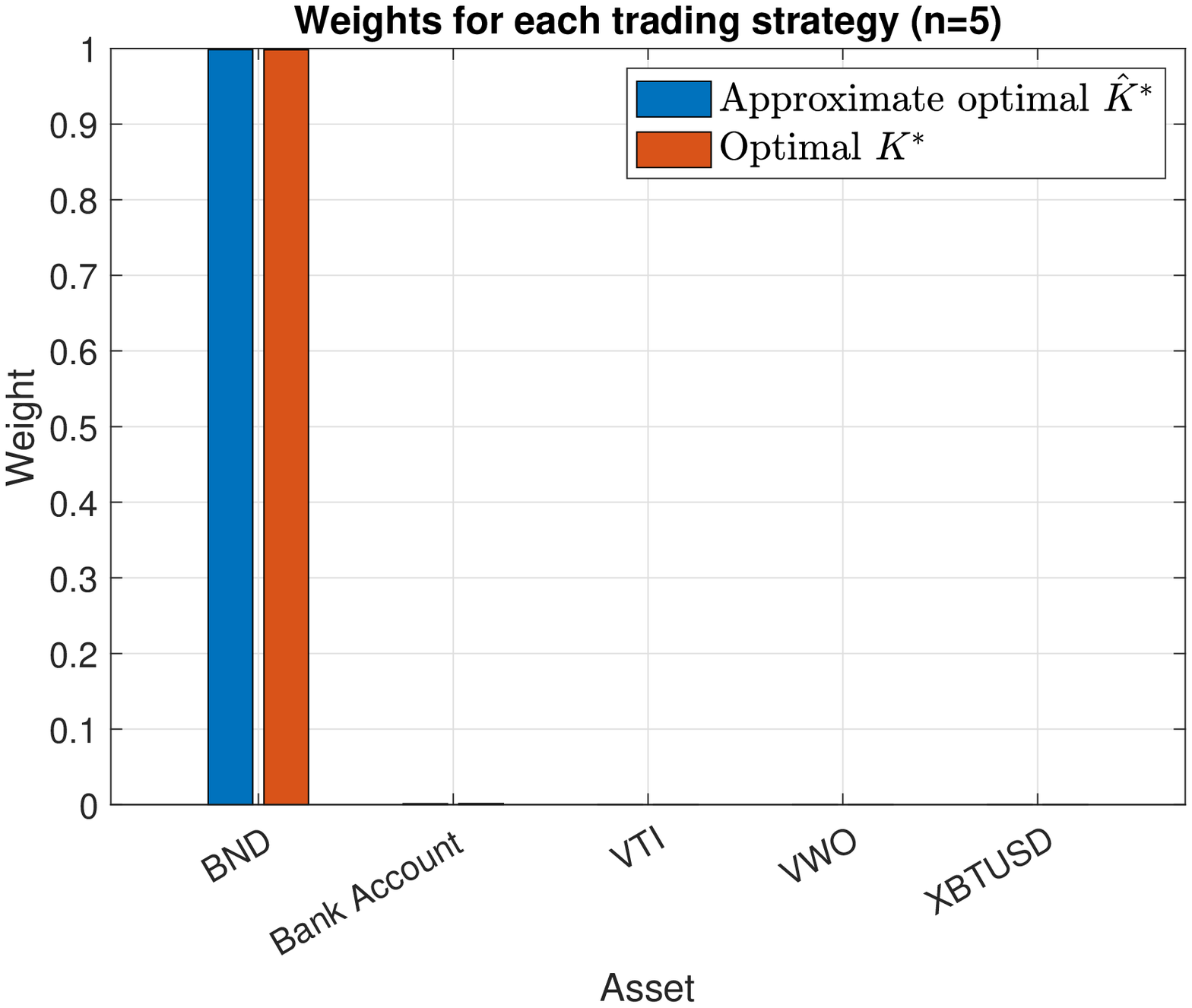}
			\caption{Portfolio Weights $K^*$ versus $\widehat{K}^*$ with  Rebalancing Period $n=5$ (Minutes).}
			\label{fig: weights (n=5)}
		\end{figure}

		
		
	\end{example}

	\medskip
	\begin{example} [Trading Performance with Different Rebalancing Periods and Costs] \rm \label{example: Trading Performance with Different Rebalancing Periods}
		This example illustrates the in-sample and out-of-sample trading performances using the solutions obtained in previous Example~\ref{example: The Similarity of K* and KHatStar}.
		Specifically, let $V(N)$ be the account value at the terminal stage $N$.
		The \textit{portfolio realized return} in period $k$ is~$R^p(k) := \frac{V(k+1) - V(k)}{V(k)}$.
		With the aid of this realized return, we consider the following metrics to study the trading performance:  The realized \textit{cumulative rate of return} $\frac{V(N)-V(0)}{V(0)}$, \textit{realized log-return} $\log \frac{V(N)}{V(0)}$,  \textit{volatility} $\sigma:= {\rm std}(R^p(k))$, \textit{maximum percentage drawdown}~$d^*:= \max_{0\leq k \leq N} \frac{V_{\max}(k) - V(k)}{V_{\max}(k)}$ with $V_{\max}(k):= \max_{0 \leq i \leq k} V(i)$, and the~\textit{$N$-period Sharpe ratio} $\sqrt{N} \cdot SR$ with $SR$ being the \textit{per-period realized Sharpe ratio}.\footnote{Given a  sequence of the realized  portfolio per-period returns $\{R^p(k): k=0,1,\dots,N-1\}$, the per-period Sharpe ratio is $SR:= \frac{\overline{R}^p - r_f}{s}$ where $\overline{R}^p:=\frac{1}{N} \sum_{k=0}^{N-1} R^p(k)$ is the sample mean return, $r_f$ is the per-period risk-free rate, and $s:= \sqrt{ \frac{1}{N-1} \sum_{k=0}^{N-1} (R^p(k) - \overline{R}^p)^2}$ is the sample standard deviation of portfolio returns. A detailed discussion of this topic can be found in \cite{lo2002statistics}.}
		
		Starting with initial account $V(0)= \$1$, Figures~\ref{fig: trading performance(n=1)} and~\ref{fig: trading performance(n=5)} reveal the in-sample and out-of-sample values of the trading account using the three trading strategies: The log-optimal portfolio with weight~$K^*$, the approximate log-optimum~$\widehat{K}^*$, and buy-and-hold with equal weight~$K=1/m$, for the same five-asset portfolio considered in Example~\ref{example: two fund thm for five asset portfolio}. 
		Note that there are nonzero transaction costs of $0.001\%$ for trading ETFs and a cost of $0.1\%$ for trading cryptocurrency. 
		From the figures, we see that the account value trajectory obtained using~$\widehat{K}^*$ is similar to that obtained using~$K^*$.
		Moreover, both of the portfolios outperform the equally-weighted buy-and-hold strategy. 
		
		To see clearly the effect of transaction costs on trading performance, we consider an additional scenario with zero costs for trading both ETFs and cryptocurrency; see Figures~\ref{fig: trading performance_ZeroCost(n=1)} and~\ref{fig: trading performance_ZeroCost(n=5)} for the  in-sample and out-of-sample account value trajectories under rebalancing period~$n=1$ and $n=5$, with zero costs. 
		Both figures demonstrate that the account values are improved when there are no costs.

		Tables~\ref{table: Descriptive Statistics} and~\ref{table: Descriptive Statistics_n=5} provide an overview of the out-of-sample  trading performance metrics of the three trading strategies for different rebalancing periods~$n=1$ and $n=5$, respectively.
		For the case of~$n=1$, i.e., the portfolio is rebalanced every minute, we find that the zero costs lead to better performance of the log-optimal portfolio in terms of the Sharpe ratio. 
		When nonzero transaction costs are imposed, the Sharpe ratios for $K^*$ and~$\widehat{K}^*$ become negative. 
		This suggests that transaction costs have a negative impact on trading performance especially when rebalancing occurs frequently. 
		On the other hand, for the case of $n=5$, where the portfolio is rebalanced every five minutes, the Sharpe ratios are positive and generally higher than those for $n=1$. This indicates that a longer rebalancing period incurs fewer  costs and may lead to better trading performance.
	\end{example}

	\medskip
	\begin{figure}[h!]
		\centering
		\includegraphics[width=.8\linewidth]{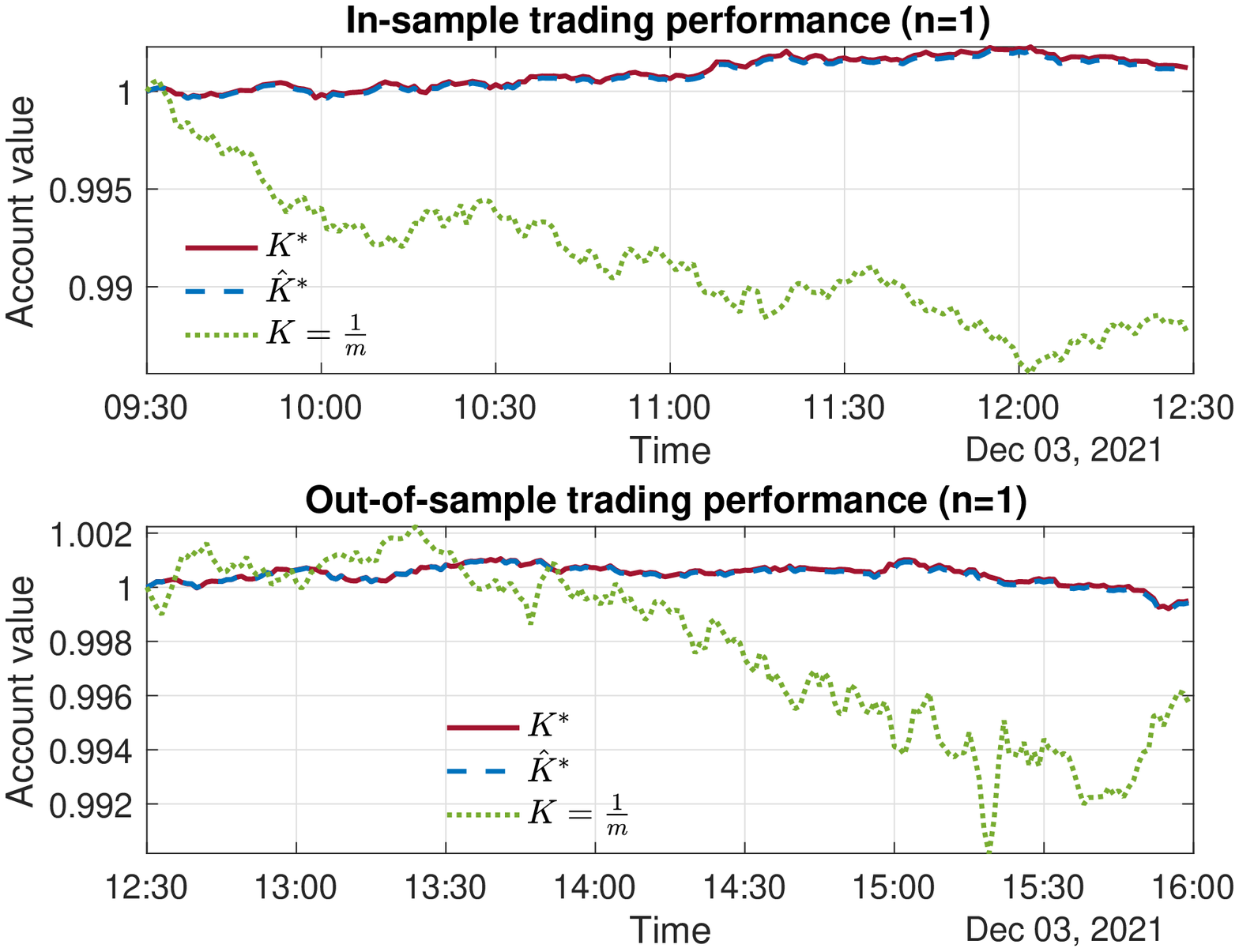}
		\caption{Account Value Trajectories under Three Trading Strategies (Optimal $K^*$, $\widehat{K}^*$, and Equally-Weighted $K=1/m$) with Rebalancing Period $n=1$ (Minute) and Nonzero Costs.}
		\label{fig: trading performance(n=1)}
	\end{figure}
	
	\medskip
	\begin{figure}[h!]
		\centering
		\includegraphics[width=.8\linewidth]{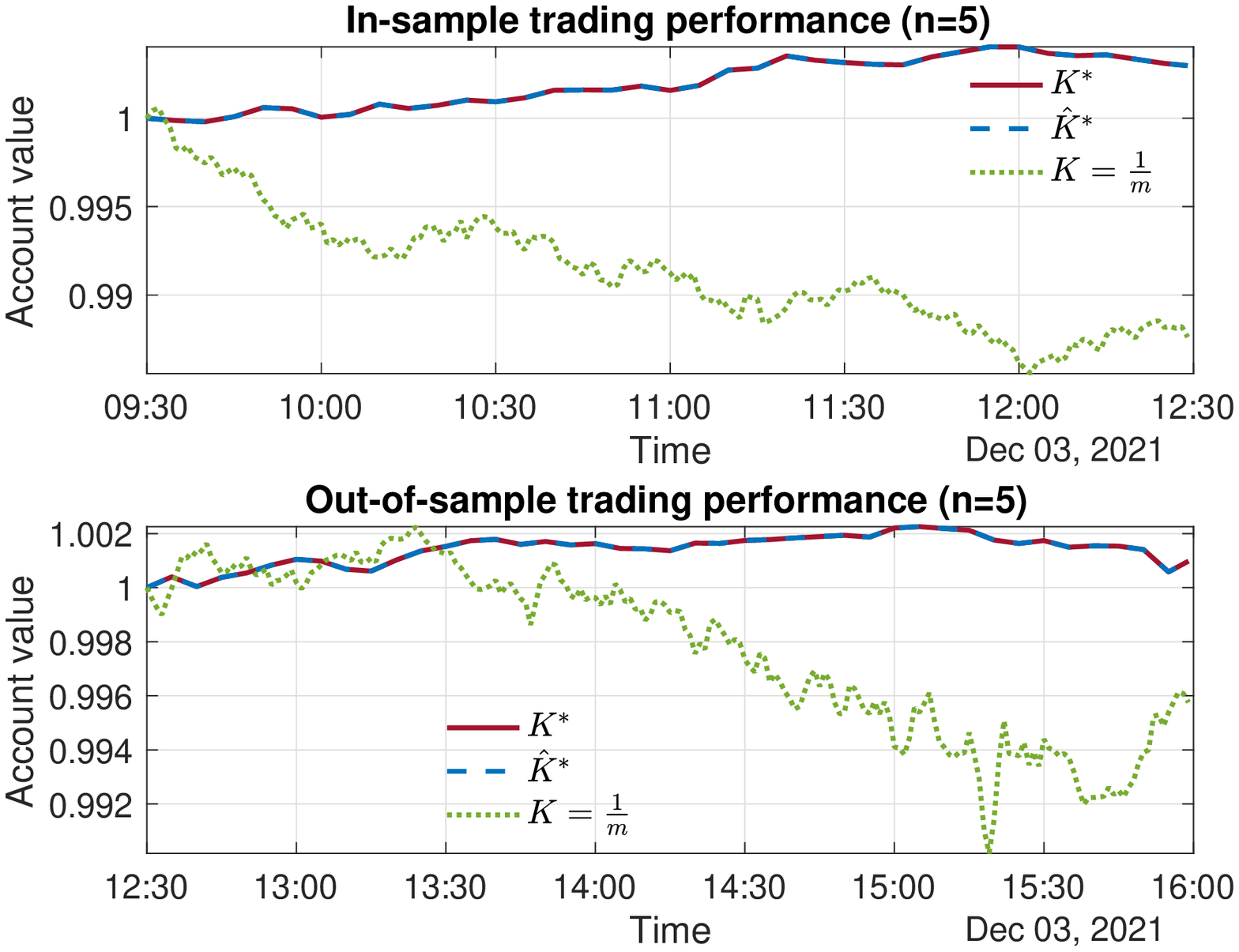}
		\caption{Account Value Trajectories under Three Trading Strategies (Optimal $K^*$, $\widehat{K}^*$, and Equally-Weighted $K=1/m$) with Rebalancing Period $n=5$ (Minutes)  and Nonzero Costs.}
		\label{fig: trading performance(n=5)}
	\end{figure}
	
	\medskip
	\begin{figure}[h!]
		\centering
		\includegraphics[width=.8\linewidth]{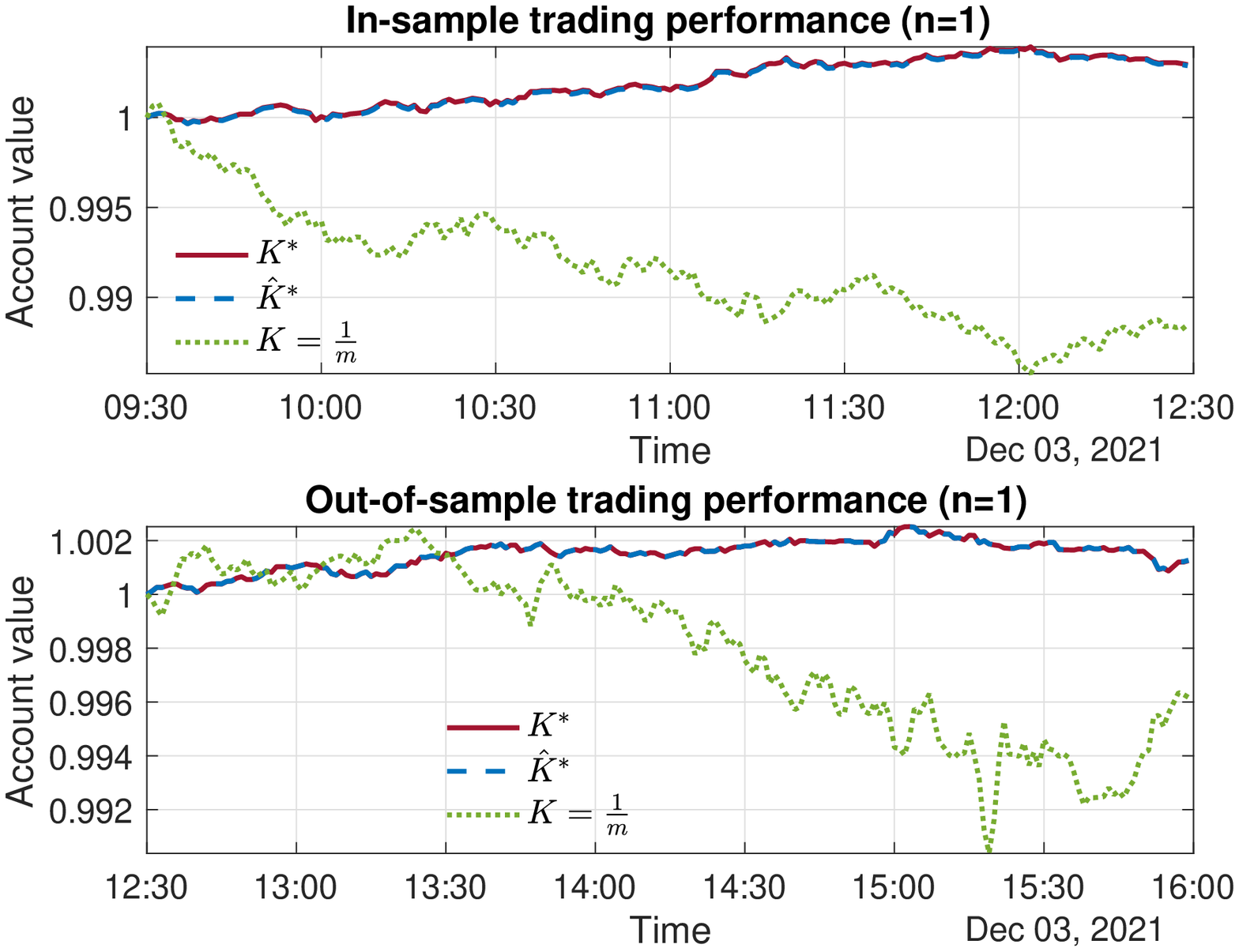}
		\caption{Account Value Trajectories under Three Trading Strategies (Optimal $K^*$, $\widehat{K}^*$, and Equally-Weighted $K=1/m$) with Rebalancing Period $n=1$ (Minute) and Zero Costs.}
		\label{fig: trading performance_ZeroCost(n=1)}
	\end{figure}
	
	\medskip
	\begin{figure}[h!]
		\centering
		\includegraphics[width=.8\linewidth]{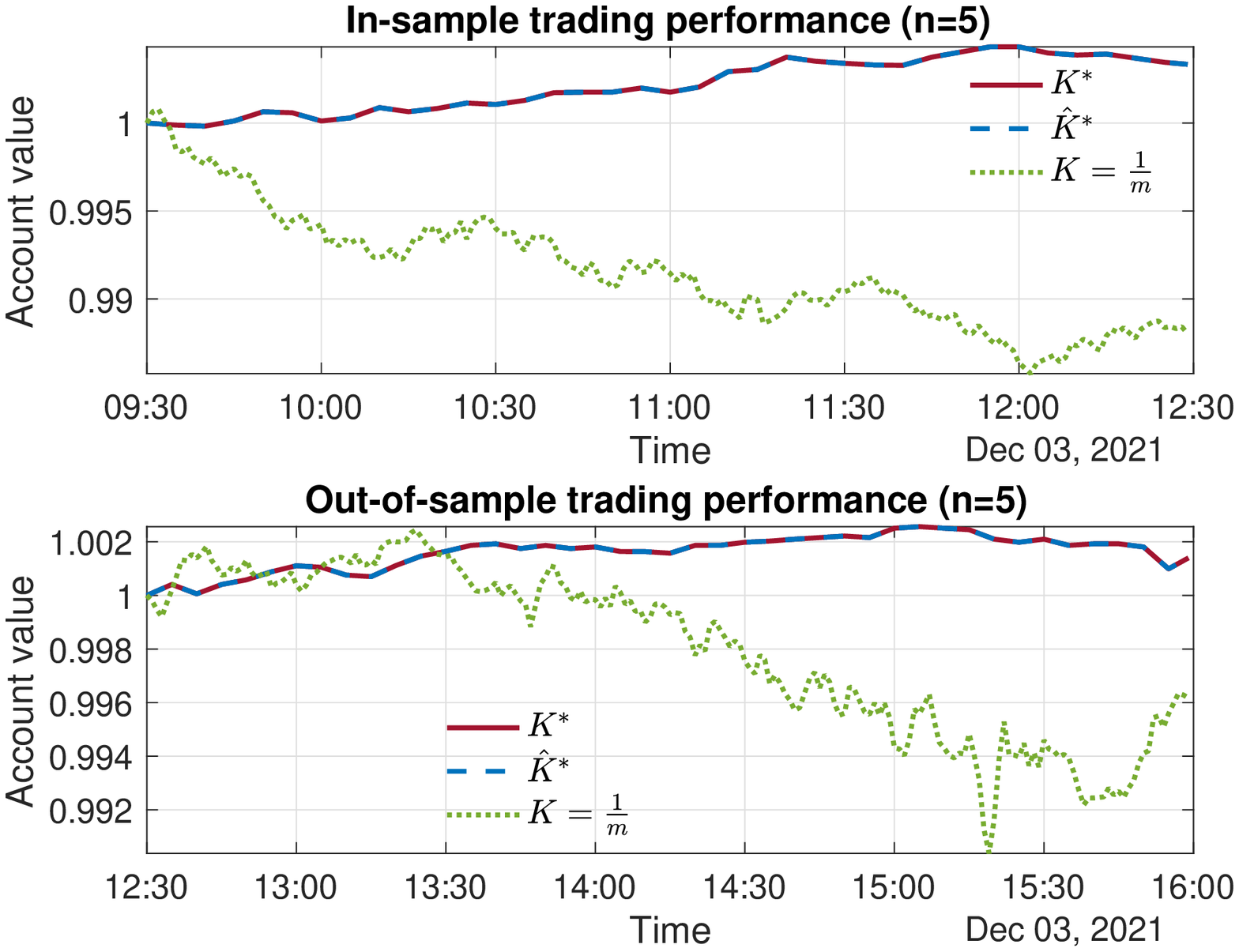}
		\caption{Account Value Trajectories under Three Trading Strategies (Optimal $K^*$, $\widehat{K}^*$, and Equally-Weighted $K=1/m$) with Rebalancing Period $n=5$ (Minutes) and Zero Costs.}
		\label{fig: trading performance_ZeroCost(n=5)}
	\end{figure}
	
	\medskip
	\begin{table*}[h!] 
		\centering
		\caption{Out-of-Sample Trading Performance Metrics with Different Transaction Costs with Rebalancing Period~$n=1$ (Minute)} 
		\begin{tabular}{|c|c|c|c|}
			\hline 
			Costs of $0\%$  for ETFs and cryptocurrency & $K^*$ & $\widehat{K}^*$ & Buy and hold \\
			\hline
			Cumulative rate of return $\frac{V(N)-V(0)}{V(0)}$ $(\%)$ & $0.13$ & $0.13$ & $-0.38$  \\
			\hline
			Realized log-growth $\log\frac{V(N)}{V(0)}$ $(\%)$ & $0.13$ & $0.13$ & $-0.38$ \\
			\hline
			Volatility $\sigma$ $(\%)$ & $0.01$ & $0.01$ & $0.05$ \\
			\hline
			Maximum percentage drawdown $d^*$ $(\%)$ & $0.16$ & $0.16$ & $1.20$\\
			\hline
			Sharpe ratio $\sqrt{N}SR$ &  $0.74$ &  $0.75$ &  $-0.58$ \\
			\hline
			\hline
			\hline
			Costs of $0.001\%$ for ETFs and $0.1\%$ for cryptocurrency & $K^*$ & $\widehat{K}^*$ & Buy and hold \\
			\hline
			Cumulative rate of return $\frac{V(N)-V(0)}{V(0)}$ $(\%)$ & $-0.05$ & $-0.06$ & $-0.42$  \\
			\hline
			Realized log-growth $\log\frac{V(N)}{V(0)}$ $(\%)$ & $-0.05$ & $-0.06$ & $-0.43$\\
			\hline
			Volatility $\sigma$ $(\%)$ & $0.01$ & $0.01$ & $0.05$\\
			\hline
			Maximum percentage drawdown $d^*$ $(\%)$ & $0.18$ & $0.19$ & $1.20$\\
			\hline
			Sharpe ratio $\sqrt{N}SR$ &  $-0.58$ &  $-0.65$ &  $-0.64$ \\
			\hline
			\hline
		\end{tabular}
		\label{table: Descriptive Statistics}
	\end{table*}

	\medskip
	\begin{table*}[h!] 
		\centering
		\caption{Out-of-Sample Trading Performance Metrics with Different Transaction Costs with Rebalancing Period~$n=5$ (Minutes)} 
		\begin{tabular}{|c|c|c|c|}
			\hline 
			Costs of $0\%$ for ETFs and cryptocurrency & $K^*$ & $\widehat{K}^*$ & Buy and hold \\
			\hline
			Cumulative rate of return $\frac{V(N)-V(0)}{V(0)}$ $(\%)$ & $0.14$ & $0.14$ & $-0.38$  \\
			\hline
			Realized log-growth $\log\frac{V(N)}{V(0)}$ $(\%)$ & $0.14$ & $0.14$ & $-0.38$ \\
			\hline
			Volatility $\sigma$ $(\%)$ & $0.02$ & $0.02$ & $0.05$ \\
			\hline
			Maximum percentage drawdown $d^*$ $(\%)$ & $0.16$ & $0.16$ & $1.20$\\
			\hline
			Sharpe ratio $\sqrt{N}SR$ &  $0.88$ &  $0.88$ &  $-0.58$ \\
			\hline
			\hline
			\hline
			Costs of $0.001\%$ for ETFs and $0.1\%$ for cryptocurrency & $K^*$ & $\widehat{K}^*$ & Buy and hold \\
			\hline
			Cumulative rate of return $\frac{V(N)-V(0)}{V(0)}$ $(\%)$ & $0.10$ & $0.10$ & $-0.42$  \\
			\hline
			Realized log-growth $\log\frac{V(N)}{V(0)}$ $(\%)$ & $0.10$ & $0.10$ & $-0.43$\\
			\hline
			Volatility $\sigma$ $(\%)$ & $0.02$ & $0.02$ & $0.05$\\
			\hline
			Maximum percentage drawdown $d^*$ $(\%)$ & $0.17$ & $0.17$ & $1.20$\\
			\hline
			Sharpe ratio $\sqrt{N}SR$ &  $0.61$ &  $0.61$ &  $-0.64$ \\
			\hline
			\hline
		\end{tabular}
		\label{table: Descriptive Statistics_n=5}
	\end{table*}

	\begin{example}[Mid-Sized Portfolio: Thirty-Two Assets with Daily Price Data] \rm \label{example: Robustness check}
		Our theory is readily applied to a mid-sized (or large-sized) portfolio.
		As an example,  we consider a portfolio consisting of 32 assets involving a bank account, Dow 30 Stocks,\footnote{Dow 30 Stocks consist of the thirty stocks that make up the Dow Jones Industrial Average.} and the Bitcoin-to-USD exchange rate~(Ticker:~BTC-USD) over a one-year horizon from November 20, 2021 to November 20, 2022.\footnote{The data considered in this example is retrieved from Yahoo Finance. It is worth noting that the time period considered for this example is significant because the third-largest cryptocurrency exchange, FTX, declared bankruptcy on November 11, 2022, which had a significant  impact on cryptocurrency markets.} 
		
		The one-year data is divided into two parts: The first $90$ days are used for in-sample optimization and the remainder is used for out-of-sample testing. 
		Here, we consider different scenarios for the transaction costs: zero costs,~$0.01\%$,~$0.1\%$, and $0.5\%$ for  trading stocks, and zero costs and $0.1\%$ fees for trading cryptocurrency. 
		If investors retain their capital in the bank account, they earn daily interest with a rate $r_f :=1\%/365$.
		
		Fix $n=1$, i.e., the portfolio is rebalanced on a daily basis.
		When costs for trading stocks are~$0\%$, $0.01\%$ and $0.1\%$, we find that $K_{CVX}^* \approx \widehat{K}_{CVX}^*  \approx 1$.\footnote{Note that, in this example, Chevron Corporation (Ticker: CVX) is the dominant asset since the estimated dominance condition $\max_{1\leq i \leq 32,\; i \neq CVX} \frac{1}{N}\sum_{k=1}^{N}\frac{1+X_{i}(k)}{1+X_{CVX}(k)} = 0.998 < 1$ for all the assets in the portfolio except for CVX. 
			Hence, according to Lemma~\ref{lemma: dominance theorem with costs}, a log-optimal investor should invest all the available capital in this asset.} 
		However, when the proportional cost is~$0.5\%$, the approximate optimum becomes $K^*_{\rm Bank \; account} \approx \widehat{K}^*_{\rm Bank \; account} \approx 1$, indicating that it is optimal to hold all capital in the bank account.
		Table~\ref{table: Descriptive Statistics Daily cost (n=1)} summarizes the performance of the three trading strategies under different levels of costs for trading stocks and cryptocurrency.  
		As expected, higher costs result in a significant decrease in investor revenue. 
		The corresponding account value trajectories are plotted in Figure~\ref{fig: trading performance_Daily(n=1)}.

		Subsequently, we examine the effects of different rebalancing periods by setting the rebalancing period to every five days, i.e.,  $n=5$. 
		In this case, we find that $K_{CVX}^* \approx \widehat{K}_{CVX}^* \approx 1$ for all four different levels of costs ($0\%$, $0.01\%, 0.1\%$, and $0.5\%$) for trading stocks. 
		This is in contrast to the case with $n=1$, where the optimal weights dictated $K^*_{\rm bank \; account} \approx 1$ when the proportional cost was~$0.5\%$.
		Figure~\ref{fig: trading performance_Daily(n=5)} shows that the associated trading performance using $K^*$ and $ \widehat{K}^*$ are similar and outperforms the buy-and-hold strategy with equal weights $1/m$ over the given time period. 
		Table~\ref{table: Descriptive Statistics Daily cost (n=5)} provides a summary of the performance metrics under four different levels of costs with rebalancing periods~$n=5$.
		
		For an even longer rebalancing period, say $n=10$ and $n=30$, the optimal weight $K^*_{CVX}=1$ remains under proportional cost for stocks being $0.5\%$.
	\end{example}

	
	
	\begin{figure}[h!]
		\centering
		\includegraphics[width=.8\linewidth]{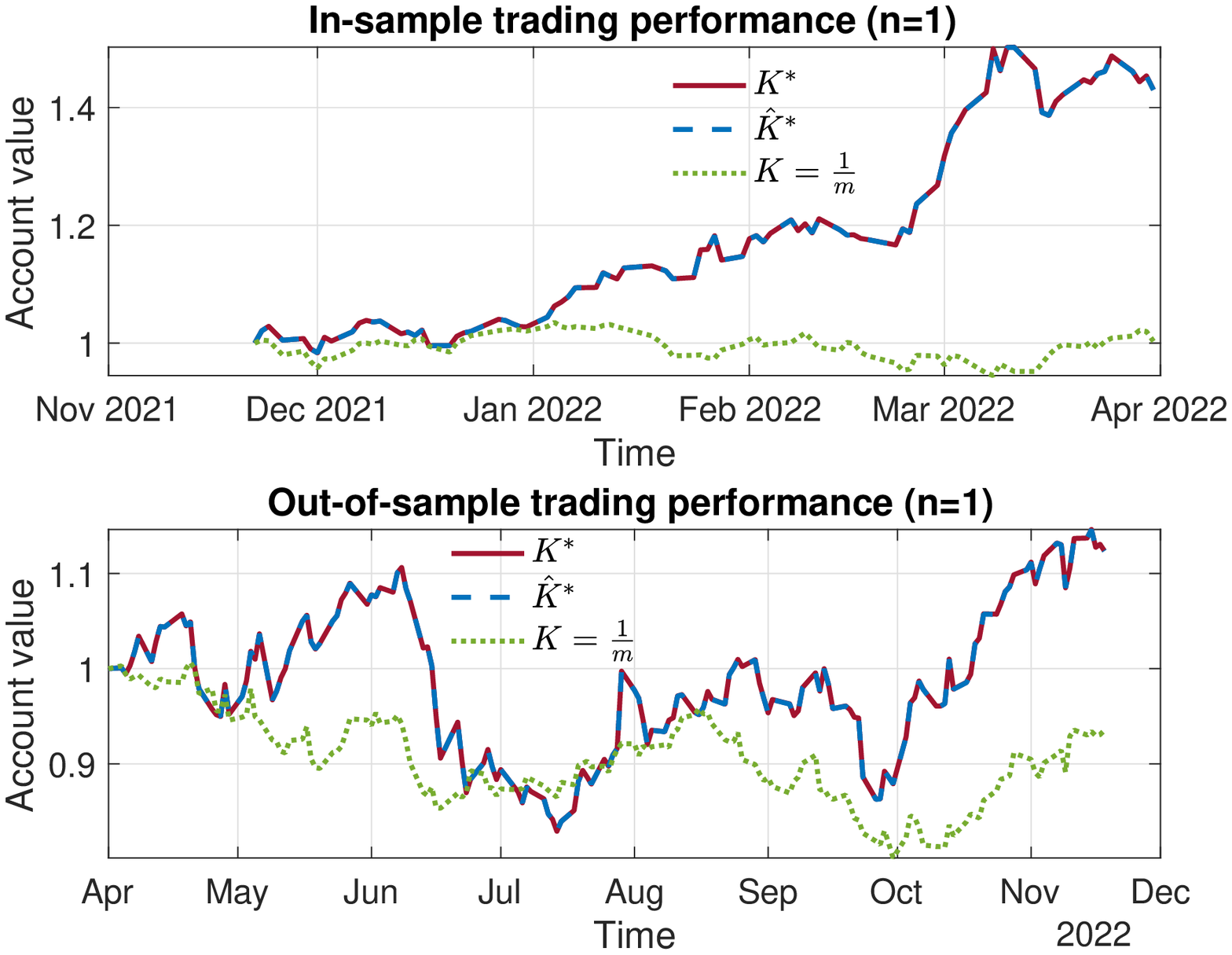}
		\caption{Account Value Trajectories under Three Trading Strategies (Optimal $K^*$, $\widehat{K}^*$, and Equally-Weighted $K=1/m$) with Rebalancing Period $n=1$ (Day) and Costs of $0.01\%$ for Stocks and $0.1\%$ for Cryptocurrency.}
		\label{fig: trading performance_Daily(n=1)}
	\end{figure}
	
	\begin{table*}[h!] 
		\centering
		\caption{Out-of-Sample Trading Performance with Zero Costs and Different Nonzero Costs for Stocks and 
			Cryptocurrency with Rebalancing Period $n=1$ (Day).} 
		\begin{tabular}{|c|c|c|c|}
			\hline
			Costs of $0\%$ for stocks and cryptocurrency & $K^*$ & $\widehat{K}^*$ & Buy and hold \\
			\hline
			Cumulative rate of return $\frac{V(N)-V(0)}{V(0)}$ $(\%)$ & $14.20$ & $14.18$ & $-6.47$ \\
			\hline
			Realized log-growth $\log\frac{V(N)}{V(0)}$ $(\%)$ & $13.28$ & $13.26$ & $-6.68$\\
			\hline
			Volatility $\sigma$ $(\%)$ & $2.20$ & $2.20$ & $1.37$\\
			\hline
			Maximum percentage drawdown $d^*$ $(\%)$ & $24.88$ & $24.89$ & $20.42$ \\
			\hline
			Sharpe ratio $\sqrt{N}SR$ &  $0.60$ &  $0.60$ &  $-0.33$ \\
			\hline
			\hline
			\hline
			Costs of $0.01\%$ for stocks and $0.1\%$ for cryptocurrency & $K^*$ & $\widehat{K}^*$ & Buy and hold \\
			\hline
			Cumulative rate of return $\frac{V(N)-V(0)}{V(0)}$ $(\%)$ & $12.39$ & $12.37$ & $-6.49$ \\
			\hline
			Realized log-growth $\log\frac{V(N)}{V(0)}$ $(\%)$ & $11.68$ & $11.66$ & $-6.71$\\
			\hline
			Volatility $\sigma$ $(\%)$ & $2.20$ & $2.20$ & $1.37$\\
			\hline
			Maximum percentage drawdown $d^*$ $(\%)$ & $25.06$ & $25.07$ & $20.42$ \\
			\hline
			Sharpe ratio $\sqrt{N}SR$ &  $0.54$ &  $0.54$ &  $-0.33$ \\
			\hline
			\hline
			\hline
			Costs of $0.1\%$ for stocks and $0.1\%$ for cryptocurrency  & $K^*$ & $\widehat{K}^*$ & Buy and hold \\
			\hline
			Cumulative rate of return $\frac{V(N)-V(0)}{V(0)}$ $(\%)$ & $-2.68$ & $-2.7$ & $-6.65$\\
			\hline
			Realized log-growth $\log\frac{V(N)}{V(0)}$ $(\%)$ & $-2.72$ & $-2.73$ & $-6.88$ \\
			\hline
			Volatility $\sigma$ $(\%)$ & $2.20$ & $2.20$ & $1.37$ \\
			\hline
			Maximum percentage drawdown $d^*$ $(\%)$ & $27.15$ & $27.17$ & $20.42$ \\
			\hline
			Sharpe ratio $\sqrt{N}SR$ &  $0.03$ &  $0.02$ &  $-0.34$ \\
			\hline
			\hline
			\hline
			Costs of $0.5\%$ for stocks and $0.1\%$ for cryptocurrency & $K^*$ & $\widehat{K}^*$ & Buy and hold \\
			\hline
			Cumulative rate of return $\frac{V(N)-V(0)}{V(0)}$ $(\%)$ & $0.22$ & $0.17$ & $-7.34$  \\
			\hline
			Realized log-growth $\log\frac{V(N)}{V(0)}$ $(\%)$ & $0.22$ & $0.17$ & $-7.63$\\
			\hline
			Volatility $\sigma$ $(\%)$ & $3.9\times10^{-5}$ & $4.7\times10^{-5}$ & $1.37$ \\
			\hline
			Maximum percentage drawdown $d^*$ $(\%)$  & $0.02$ & $0.03$ & $20.42$ \\
			\hline
			Sharpe ratio $\sqrt{N}SR$ &  $-4.45$ & $-4.55$ &  $-0.38$ \\
			\hline
			\hline
		\end{tabular}
		\label{table: Descriptive Statistics Daily cost (n=1)}
	\end{table*}
	
	
	\begin{figure}[h!]
		\centering
		\includegraphics[width=.8\linewidth]{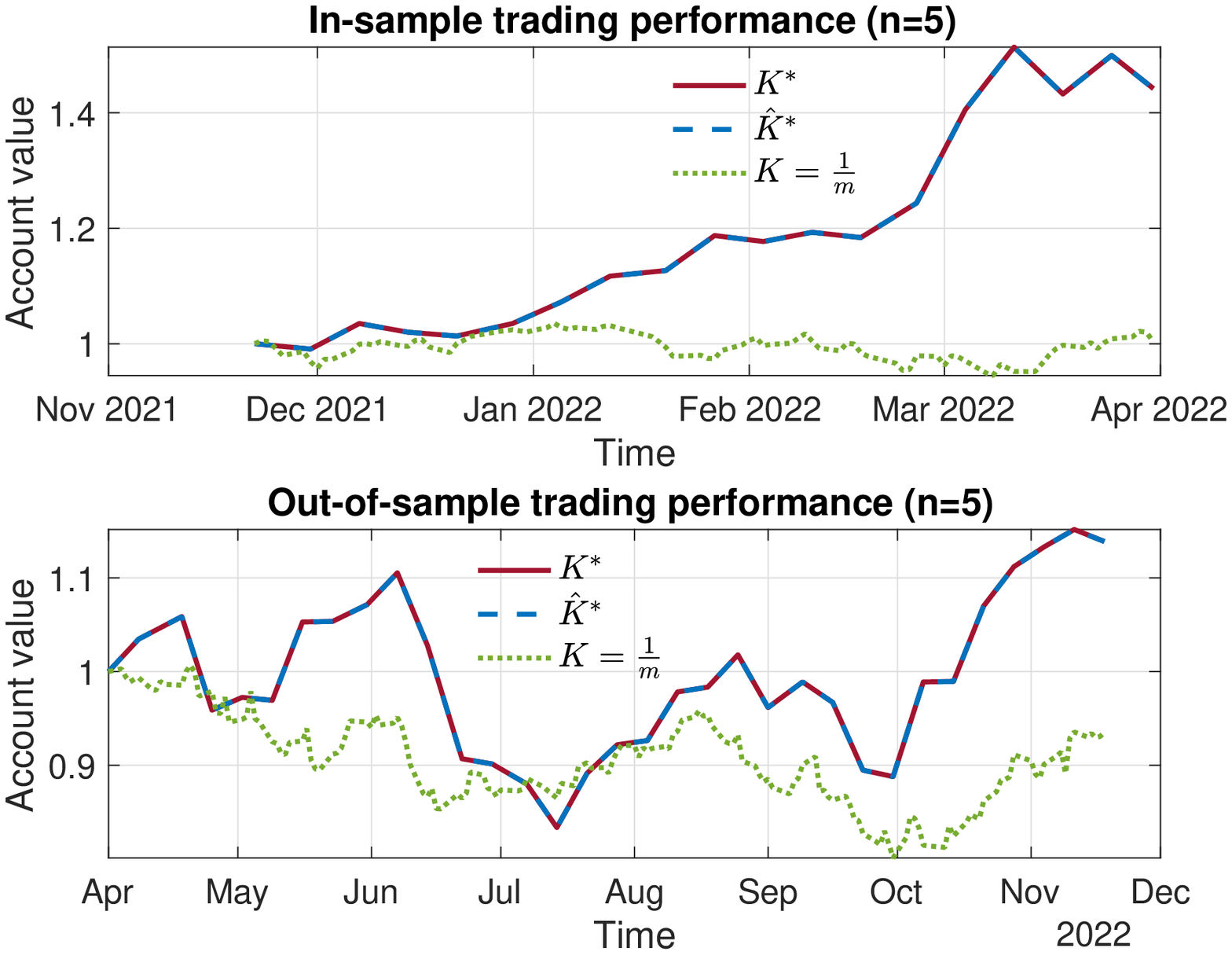}
		\caption{Account Value Trajectories under Three Trading Strategies (Optimal $K^*$, $\widehat{K}^*$, and Equally-Weighted $K=1/m$) with Rebalancing Period $n=5$ (Days)  and Costs of $0.01\%$ for Stocks and $0.1\%$ for Cryptocurrency..}
		\label{fig: trading performance_Daily(n=5)}
	\end{figure}
	
	\begin{table*}[h!] 
		\centering
		\caption{Out-of-Sample Trading Performance  with Zero Costs and Different Nonzero Costs for Stocks and
			Cryptocurrency with Rebalancing Period $n=5$ (Days).} 
		\begin{tabular}{|c|c|c|c|}
			\hline
			Costs of $0\%$ for stocks and cryptocurrency & $K^*$ & $\widehat{K}^*$ & Buy and hold \\
			\hline
			Cumulative rate of return $\frac{V(N)-V(0)}{V(0)}$ $(\%)$ & $14.25$ & $14.23$ & $-6.47$ \\
			\hline
			Realized log-growth $\log\frac{V(N)}{V(0)}$ $(\%)$ & $13.32$ & $13.31$ & $-6.68$\\
			\hline
			Volatility $\sigma$ $(\%)$ & $5.18$ & $5.18$ & $1.37$\\
			\hline
			Maximum percentage drawdown $d^*$ $(\%)$ & $24.55$ & $24.55$ & $20.42$ \\
			\hline
			Sharpe ratio $\sqrt{N}SR$ &  $0.61$ &  $0.60$ &  $-0.33$ \\
			\hline
			\hline
			\hline
			Costs of $0.01\%$ for stocks and $0.1\%$ for cryptocurrency & $K^*$ & $\widehat{K}^*$ & Buy and hold \\
			\hline
			Cumulative rate of return $\frac{V(N)-V(0)}{V(0)}$ $(\%)$ & $13.88$ & $13.87$ & $-6.49$ \\
			\hline
			Realized log-growth $\log\frac{V(N)}{V(0)}$ $(\%)$ & $13.00$ & $12.99$ & $-6.71$\\
			\hline
			Volatility $\sigma$ $(\%)$ & $5.18$ & $5.18$ & $1.37$\\
			\hline
			Maximum percentage drawdown $d^*$ $(\%)$ & $24.59$ & $24.59$ & $20.42$ \\
			\hline
			Sharpe ratio $\sqrt{N}SR$ &  $0.59$ &  $0.59$ &  $-0.33$ \\
			\hline
			\hline
			\hline
			Costs of $0.1\%$ for stocks and $0.1\%$ for cryptocurrency & $K^*$ & $\widehat{K}^*$ & Buy and hold \\
			\hline
			Cumulative rate of return $\frac{V(N)-V(0)}{V(0)}$ $(\%)$ & $10.66$ & $10.64$ & $-6.65$ \\
			\hline
			Realized log-growth $\log\frac{V(N)}{V(0)}$ $(\%)$ & $10.13$ & $10.12$ & $-6.88$ \\
			\hline
			Volatility $\sigma$ $(\%)$ & $5.18$ & $5.18$ & $1.37$ \\
			\hline
			Maximum percentage drawdown $d^*$ $(\%)$ & $24.95$ & $24.95$ & $20.42$ \\
			\hline
			Sharpe ratio $\sqrt{N}SR$ &  $0.49$ &  $0.49$ &  $-0.34$ \\
			\hline
			\hline
			\hline
			Costs of $0.5\%$ for stocks and $0.1\%$ for cryptocurrency &$K^*$ & $\widehat{K}^*$ & Buy and hold \\
			\hline
			Cumulative rate of return $\frac{V(N)-V(0)}{V(0)}$ $(\%)$ & $-2.64$ & $-2.65$ & $-7.34$  \\
			\hline
			Realized log-growth $\log\frac{V(N)}{V(0)}$ $(\%)$ & $-2.67$ & $-2.69$ & $-7.63$\\
			\hline
			Volatility $\sigma$ $(\%)$ & $5.18$ & $5.18$ & $1.37$ \\
			\hline
			Maximum percentage drawdown $d^*$ $(\%)$  & $26.53$ & $26.53$ & $20.42$ \\
			\hline
			Sharpe ratio $\sqrt{N}SR$ &  $0.05$ &  $0.05$ &  $-0.38$ \\
			\hline
			\hline
		\end{tabular}
		\label{table: Descriptive Statistics Daily cost (n=5)}
	\end{table*}

	\section{Online Trading with Sliding Window Approach} \label{section: online trading with sliding window approach}
	In previous sections, optimal weights $K^*$ and its approximation counterpart $\widehat{K}^*$ were obtained as fixed values based on the empirical distributions of returns, rather than true distribution, which is typically unknown to the investor in practice. 
	Moreover, these fixed weights cannot adapt to the constantly changing information in a dynamic market. 
	To address this issue, we apply a data-driven sliding window approach that generates time-varying log-optimal weights online; see also~\cite{wang2022data} for a similar idea for online trading.
	
	The idea of the sliding window approach is as follows. For $k=0,1,\dots$, the investor first declares a fixed window size $M \geq 1$. With $k=0,1, \dots, M-1$, one solves the log-optimal portfolio problem~(\ref{problem: simplified log-optimization problem}) to obtain $K^*$ or the approximation counterpart (\ref{problem: approximate log-optimal portfolio problem}) to obtain $\widehat{K}^*$. 
	These optimum weights are then applied in the next stage.
	Having done that, one re-solves the log-optimal portfolio problem again using the data from $k=1,2,\dots, M$. 
	Repeating this procedure until the end, one obtains a time-varying optimum $K^*(k)$ or $\widehat{K}^*(k).$ 
	This approach has a computational advantage because it solves a sequence of concave optimization problems rather than a stochastic dynamic programming problem.
	The details of this approach can be found in 
	Algorithm~\ref{algorithm: sliding window approach for log-optimum} below.

	\begin{algorithm}
		\caption{Online Trading via Sliding Window Approach}
		\begin{algorithmic}[1]
			\Require Consider $m \geq 2$ assets, realized returns $\{ X_i(k): k \geq 0\}$ and transaction cost $c_i$ for $i=1,2,\cdots,m$, and sliding window size $M \geq 1.$
			\Ensure Optimal portfolio weight $K^*$ or $\widehat{K}^*$ for each stage.
			\State Compute compound returns $\{ \widetilde{\mathcal{X}}_{n,i}(s) \}$ for each asset $i$ in the portfolio with $\widetilde{\mathcal{X}}_{n,i}(s) := \prod_{k=n(s-1)}^{ns-1}(1+X_i(k))-1-c_i$.
			
			
			\If{$k \geq M$} 
			\State Solve the maximization Problem~(\ref{problem: simplified log-optimization problem}) to obtain $K^*$ (or solve Problem~(\ref{problem: approximate log-optimal portfolio problem}) to obtain $\widehat{K}^*$).
			
			\State Having obtained optimal $K^*(k):= K^*$ (or $K^*(k):=\widehat{K}^*(k)$), we apply it at stage $k+1$. Set~$k:=k+1$ then back to Step $2$.
			\EndIf
		\end{algorithmic}
		\label{algorithm: sliding window approach for log-optimum}
	\end{algorithm}
	
	\bigskip
	\begin{example}[Mid-Sized Portfolio Revisited: Online Trading via the Sliding Window Approach] \rm
		To illustrate the sliding window approach, 
		we conduct additional empirical studies using the daily price data considered in Example~\ref{example: Robustness check} with the costs of $0.01\%$ for stocks and $0.1\%$ for cryptocurrency. 
		Here, we first fix the rebalancing period $n=1$ day and consider three different window sizes:~$M=10,20,30$ days.  
		By solving the log-optimal and approximate log-optimal portfolio problems, we obtain the resulting time-varying optimal weights~$K^*(k)$ and the approximate log-optimum~$\widehat{K}^*(k)$ for $ k=1,2, \dots$, see Figure~\ref{fig: Sliding Window Weight (n=1)} for an illustration. 
		The associated account value trajectories of three portfolios with different weights ($\widehat{K}^*$, $K^*$, and an equally-weights~$K = 1/m$) are depicted in Figure~\ref{fig: Sliding Window(n=1)}.
		See also Table~\ref{table: Descriptive Statistics Daily cost sliding window 0.01 (n=1)} for a summary of the trading performance metrics under three different window sizes $M$. 
		It is interesting to note that the portfolios with weights $K^*$ and $\widehat{K}^*$ using the window size $M=30$ outperform the buy-and-hold strategy in terms of Sharpe ratio.
		This observation suggests that the window size $M$ may be an important factor in determining the overall trading performance. 
		While this point is not pursued further in this paper, it is worth considering in future work when implementing the sliding window approach in practice.
		
		Likewise, we also study the performance with different rebalancing periods $n=5$ and $n=10$ and with different window sizes $M=10, 20,$ and $ 30$. 
		These results are summarized in Tables~\ref{table: Descriptive Statistics Daily cost sliding window 0.01 (n=5)} and~\ref{table: Descriptive Statistics Daily cost sliding window 0.01 (n=10)}. 
		Similar to the $n=1$ case, we see that for both $n=5$ and $n=10$, the best performance is obtained with $M=30$ and $M=20$, respectively in this example.

		\begin{figure}[h!]
			\centering
			\begin{subfigure}[b]{0.95\textwidth}
				\centering
				\includegraphics[width=0.95\textwidth]{"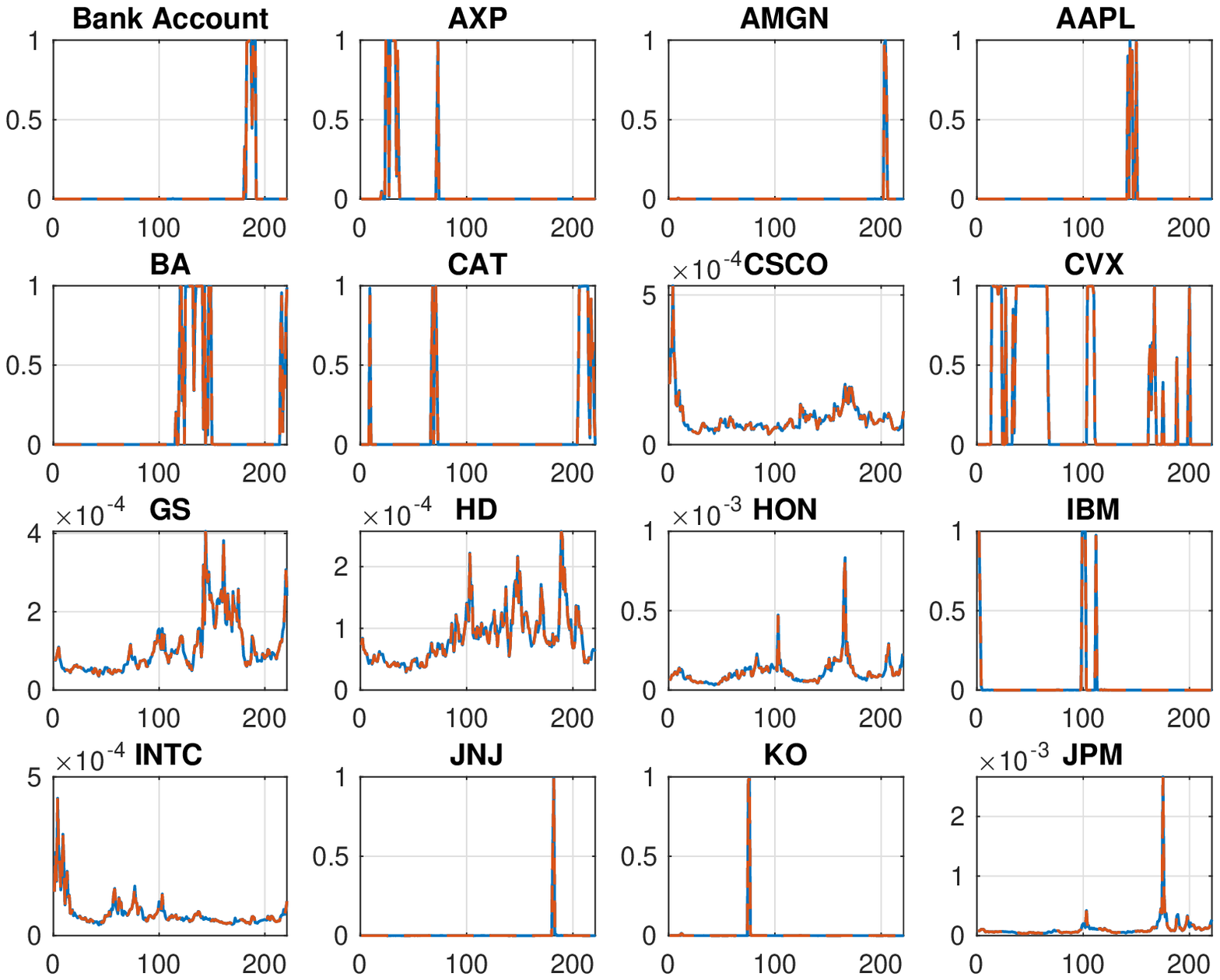"}
			\end{subfigure}
			\hfill
			\begin{subfigure}[b]{0.95\textwidth}
				\centering
				\includegraphics[width=0.95\textwidth]{"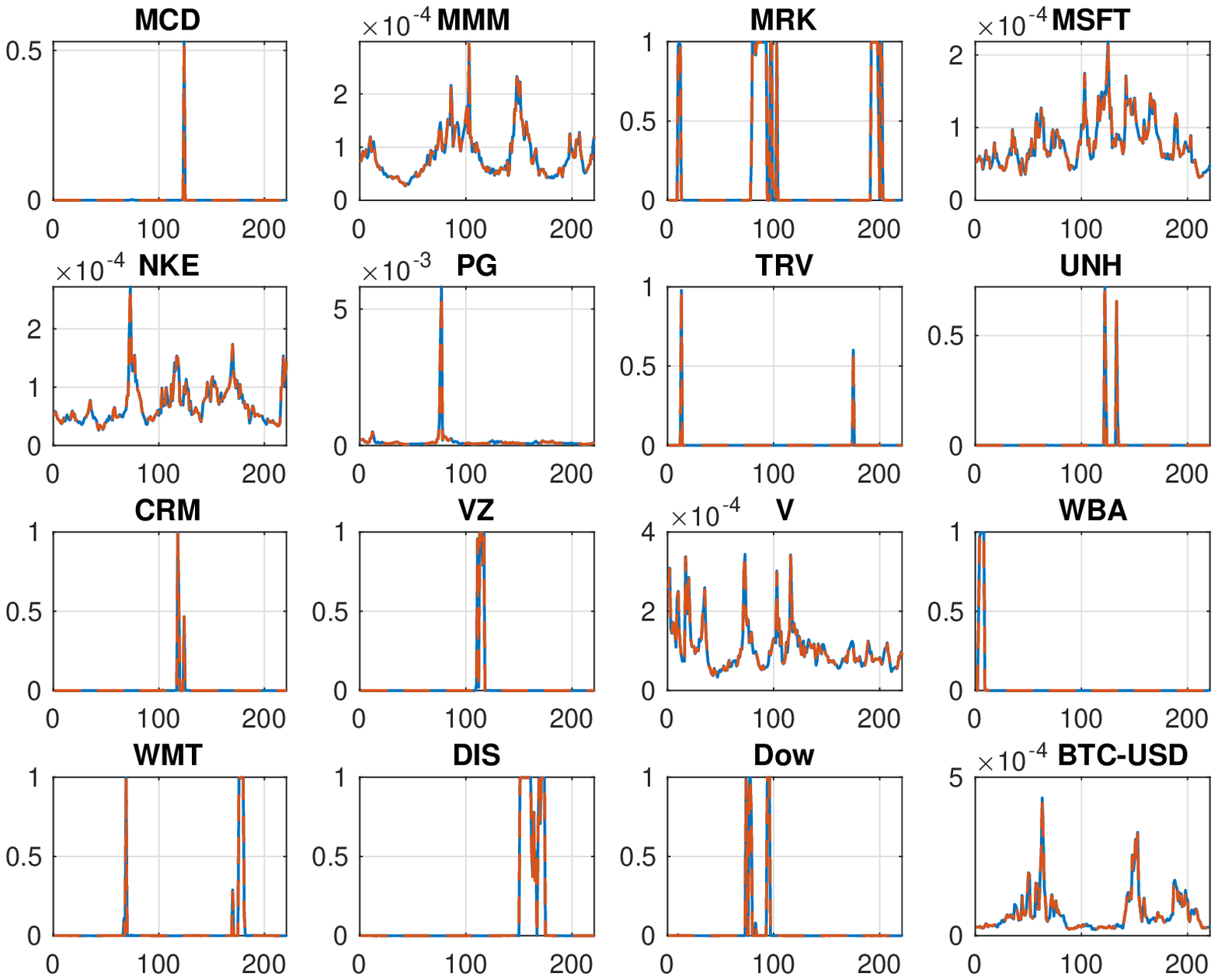"}
			\end{subfigure}
			\hfill
			\caption{Time-Varying Portfolio Weight $K^*(k)$ (red dash line) and $\widehat{K}^*(k)$ (blue solid line) with Window Size $M=30$ (Days) and Rebalancing Period $n=1$ (Day).}
			\label{fig: Sliding Window Weight (n=1)}
		\end{figure}
		
		
		
		
		
		\begin{figure}[h!]
			\centering
			\includegraphics[width=.8\linewidth]{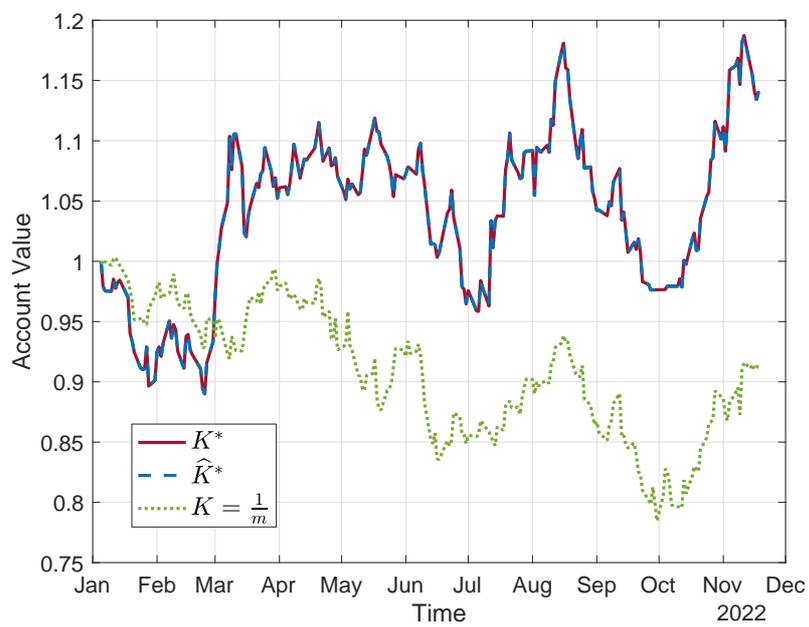}
			\caption{Equally Weighted Portfolio Versus Sliding Window Approach with Window Size $M=30$ (Days) and Rebalancing Period $n=1$ (Day).}
			\label{fig: Sliding Window(n=1)}
		\end{figure}
		
		\begin{table*}[h!] 
			\centering
			\caption{Online Trading Performance Using Sliding Window Approach with Transaction Costs of~$0.01\%$ for Stocks and $0.1\%$ for Cryptocurrency with rebalancing period $n=1$ (Day).} 
			\begin{tabular}{|c|c|c|c|}
				\hline
				$M=10$ & $K^*$ & $\widehat{K}^*$ & Buy and hold \\
				\hline
				Cumulative rate of return $\frac{V(N)-V(0)}{V(0)}$ $(\%)$ & $-22.88$ & $-23.11$ & $-5.78$ \\
				\hline
				Realized log-growth $\log\frac{V(N)}{V(0)}$ $(\%)$ & $-25.99$ & $-26.28$ & $-5.95$\\
				\hline
				Volatility $\sigma$ $(\%)$ & $2.16 $& $2.16$ & $1.24$\\
				\hline
				Maximum percentage drawdown $d^*$ $(\%)$ & $42.19$ & $42.35$ & $22.08$\\
				\hline
				Sharpe ratio $\sqrt{N}SR$ &  $-0.63$ &  $-0.62$ & $-0.25$\\
				\hline
				\hline
				\hline
				$M=20$ & $K^*$ & $\widehat{K}^*$ & Buy and hold \\
				\hline
				Cumulative rate of return $\frac{V(N)-V(0)}{V(0)}$ $(\%)$ & $-11.74$ & $-11.65$ & $-5.89$ \\
				\hline
				Realized log-growth $\log\frac{V(N)}{V(0)}$ $(\%)$ & $-12.48$ & $-12.38$ & $-6.07$\\
				\hline
				Volatility $\sigma$ $(\%)$ & $1.86$ & $1.86$ & $1.26$\\
				\hline
				Maximum percentage drawdown $d^*$ $(\%)$ & $37.85$ & $37.80$ & $22.08$\\
				\hline
				Sharpe ratio $\sqrt{N}SR$ & $ -0.33$ & $-0.30$ &  $-0.26$ \\
				\hline
				\hline
				\hline
				$M=30$ & $K^*$ & $\widehat{K}^*$ & Buy and hold \\
				\hline
				Cumulative rate of return $\frac{V(N)-V(0)}{V(0)}$ $(\%)$ & $14.12$ & $14.04$ & $-8.42$ \\
				\hline
				Realized log-growth $\log\frac{V(N)}{V(0)}$ $(\%)$ & $13.21$ & $13.14$ & $-8.79$\\
				\hline
				Volatility $\sigma$ $(\%)$ & $1.71$ & $1.71$ & $1.28$\\
				\hline
				Maximum percentage drawdown $d^*$ $(\%)$ & $17.33$ & $17.33$ & $21.80$\\
				\hline
				Sharpe ratio $\sqrt{N}SR$ &  $0.62$ &  $0.64$ &  $-0.40$ \\
				\hline
			\end{tabular}
			\label{table: Descriptive Statistics Daily cost sliding window 0.01 (n=1)}
		\end{table*}
		
		\begin{table*}[h!] 
			\centering
			\caption{Online Trading Performance Using Sliding Window Approach with Transaction Costs of~$0.01\%$ for Stocks and $0.1\%$ for Cryptocurrency with Rebalancing Period $n=5$ (Days).} 
			\begin{tabular}{|c|c|c|c|}
				\hline
				$M=10$& $K^*$ & $\widehat{K}^*$ & Buy and hold \\
				\hline
				Cumulative rate of return $\frac{V(N)-V(0)}{V(0)}$ $(\%)$ & $-2.09$ & $-2.28$ & $-7.45$ \\
				\hline
				Realized log-growth $\log\frac{V(N)}{V(0)}$ $(\%)$ & $-2.11$ & $-2.31$ & $-7.75$\\
				\hline
				Volatility $\sigma$ $(\%)$ & $4.62 $& $4.61$ & $1.33$\\
				\hline
				Maximum percentage drawdown $d^*$ $(\%)$ & $30.90$ & $31.09$ & $21.03$\\
				\hline
				Sharpe ratio $\sqrt{N}SR$ &  $0.07$ &  $0.06$ & $ -0.35 $\\
				\hline
				\hline
				\hline
				$M=20$ & $K^*$ & $\widehat{K}^*$ & Buy and hold \\
				\hline
				Cumulative rate of return $\frac{V(N)-V(0)}{V(0)}$ $(\%)$ & $-9.55$ & $-9.56$ & $-3.97$ \\
				\hline
				Realized log-growth $\log\frac{V(N)}{V(0)}$ $(\%)$ & $-10.03$ & $-10.05$ & $-4.05$\\
				\hline
				Volatility $\sigma$ $(\%)$ & $4.55$ & $4.55$ & $1.40$\\
				\hline
				Maximum percentage drawdown $d^*$ $(\%)$ & $27.58$ & $27.59$ & $18.17$\\
				\hline
				Sharpe ratio $\sqrt{N}SR$ & $ -0.29$ & $-0.29$ &  $-0.18$ \\
				\hline
				\hline
				\hline
				$M=30$ &$K^*$ & $\widehat{K}^*$ & Buy and hold \\
				\hline
				Cumulative rate of return $\frac{V(N)-V(0)}{V(0)}$ $(\%)$ & $3.50$ & $3.64$ & $6.96$ \\
				\hline
				Realized log-growth $\log\frac{V(N)}{V(0)}$ $(\%)$ & $3.44$ & $3.57$ & $6.73$\\
				\hline
				Volatility $\sigma$ $(\%)$ & $3.42$ & $3.42$ & $1.33$\\
				\hline
				Maximum percentage drawdown $d^*$ $(\%)$ & $10.54$ & $10.52$ & $16.39$\\
				\hline
				Sharpe ratio $\sqrt{N}SR$ & $0.30$ &  $0.31$ &  $0.56$ \\
				\hline
				\hline
			\end{tabular}
			\label{table: Descriptive Statistics Daily cost sliding window 0.01 (n=5)}
		\end{table*}
		
		\begin{table*}[h!] 
			\centering
			\caption{Online Trading Performance Using Sliding Window Approach with Transaction Costs of~$0.01\%$ for Stocks and $0.1\%$ for Cryptocurrency with Rebalancing Period $n=10$ (Days).} 
			\begin{tabular}{|c|c|c|c|}
				\hline
				$M=10$& $K^*$ & $\widehat{K}^*$ & Buy and hold \\
				\hline
				Cumulative rate of return $\frac{V(N)-V(0)}{V(0)}$ $(\%)$ & $-9.43$ & $-9.12$ & $-1.15$ \\
				\hline
				Realized log-growth $\log\frac{V(N)}{V(0)}$ $(\%)$ & $-9.91$ & $-9.56$ & $-1.16$\\
				\hline
				Volatility $\sigma$ $(\%)$ & $6.31$ & $6.34$ & $1.39$\\
				\hline
				Maximum percentage drawdown $d^*$ $(\%)$ & $27.56$ & $27.56$ & $18.15$\\
				\hline
				Sharpe ratio $\sqrt{N}SR$ &  $-0.30$ &  $-0.28$ &  $-0.01$ \\
				\hline
				\hline
				\hline
				$M=20$& $K^*$ & $\widehat{K}^*$ & Buy and hold \\
				\hline
				Cumulative rate of return $\frac{V(N)-V(0)}{V(0)}$ $(\%)$ & $11.32$ & $11.19$ & $11.20$ \\
				\hline
				Realized log-growth $\log\frac{V(N)}{V(0)}$ $(\%)$ & $10.73$ & $10.60$ & $10.61$\\
				\hline
				Volatility $\sigma$ $(\%)$ & $7.26$ & $7.24$ & $1.53$\\
				\hline
				Maximum percentage drawdown $d^*$ $(\%)$ & $9.51$ & $9.51$ & $4.74$\\
				\hline
				Sharpe ratio $\sqrt{N}SR$ &  $0.80$ & $0.79$ & $1.13$ \\
				\hline
				\hline
			\end{tabular}
			\label{table: Descriptive Statistics Daily cost sliding window 0.01 (n=10)}
		\end{table*}
		
	\end{example}

	\section{Concluding Remarks}\label{section: Concluding remark}
	This paper focuses on incorporating rebalancing frequency and transaction costs into the log-optimal portfolio formulation, which aims to maximize the expected logarithmic growth rate of an investor's wealth. 
	We demonstrate that solving a frequency-dependent optimization problem with costs is equivalent to solving a concave program. Conditions under which a log-optimal investor would invest all available funds in a specific asset are provided. 
	We also consider the issue of bankruptcy that can arise due to transaction costs in the frequency-dependent formulation and propose an approximate solution using a quadratic concave program.
	Additionally, a version of the two-fund theorem is proven, demonstrating that a convex combination of two optimal weights is still optimal.
	We present various empirical studies to explore the effect of considering percentage transaction cost and rebalancing periods from the small to mid-sized portfolio optimization problems. 
	Lastly, we extend our empirical studies to an online trading scenario by implementing a sliding window approach, which allows us to solve a sequence of concave programs rather than a complex stochastic dynamic programming problem.

	Regarding further research, one possible continuation is to consider additional practical trading issues; e.g., allowing to short an asset, i.e., $K_i < 0$ for some $i$ and/or modeling the impact of dividend/taxes. 
	Another feasible direction is to incorporate an \textit{risk} term into the objective function for the ELG maximization problem, which would mitigate the situation when the optimum suggests betting all capital on a specific asset; e.g., see \cite{davis2008risk}.
	Another important consideration is the potential for estimation error in the distribution of returns, which is often unknown and must be estimated in practice.
	In this case, it may be useful to study the robust counterpart of the problem considered in this paper.  That is, instead of solving $\sup_K \mathbb{E}[\log \frac{V_K(N)}{V(0)}]$, one seeks to solve a data-driven distributional robust log-optimal portfolio problem  
	$$
	\sup_{K \in \mathcal{K}} \inf_{P \in \mathcal{P}} \mathbb{E}^P \left[ \log \frac{V_K(N)}{V(0)} \right] 
	$$  where $\cal P$ is the \textit{ambiguity} set of probability distribution; e.g., see \cite{mohajerin2018data, wu2022generalization} for an approach using \textit{Wasserstein} metric to characterize the ambiguity~set.

	\bibliographystyle{apalike}
	\bibliography{refs}    
	
	\appendix
	\section{Technical Lemma on Survival Trades}

	\begin{lemma} \label{lemma: second surivial lemma}
		Let $\min_i c_i >0.$
		If $P(V(n) >0)=1$, then $\sum_{i=1}^m K_i ((1+\mu_i)^n - c_i) \geq 0.$
	\end{lemma}

	\begin{proof} 
		Fix $\min_i c_i>0$ and assume that $P(V(n)>0) = 1$. Then
		\begin{align*}
			1 = P(V(n)>0) 
			&= P\left( \sum_{i=1}^m K_i\left(   \prod_{k=0}^{n-1} (1+ X_i(k))  \right) > \sum_{i=1}^m K_i  c_i    \right).
		\end{align*}
		Since $\sum_{i=1}^m K_i  c_i >0$ and $\sum_{i=1}^m K_i\left(   \prod_{k=0}^{n-1} (1+ X_i(k))  \right) \geq 0$, applying Markov inequality yields
		\begin{align*}
			1 = P\left( \sum_{i=1}^m K_i\left(   \prod_{k=0}^{n-1} (1+ X_i(k))  \right) > \sum_{i=1}^m K_i  c_i    \right) 
			&\leq \frac{ \mathbb{E}\left[ \sum_{i=1}^m K_i\left(   \prod_{k=0}^{n-1} (1+ X_i(k)) \right) \right]}{\sum_{i=1}^m K_i c_i} \\
			&= \frac{ \sum_{i=1}^m K_i (1+ \mu_i)^n }{ \sum_{i=1}^m K_i c_i}.
		\end{align*}
		where the last equality holds since the returns $\{X_i(k): k \geq 0\}$ are i.i.d.
		Hence, by rearranging the inequality above, we obtain
		$
		\sum_{i=1}^m K_i  ((1+\mu_i)^n -c_i ) \geq 0,
		$ which is desired.
	\end{proof}
\end{document}